\documentclass[aps,prx,twocolumn,superscriptaddress,floatfix,nofootinbib]{revtex4-2}
\usepackage{subcaption}
\usepackage{graphicx}
\usepackage{quantikz}
\usepackage{amsthm}
\usepackage{hyperref}
\usepackage{cleveref}
\usepackage{mathtools}
\usepackage{tabularx}
\usepackage{booktabs}
\usepackage{tikz}
\usepackage{amsmath, amsfonts}
\usepackage{makecell}
\usepackage{blkarray}
\usepackage{float}
\usepackage{todonotes}
\usepackage{ragged2e}
\usepackage{ifthen}  

\DeclarePairedDelimiter\norm{\lVert}{\rVert}
\newcommand{\smallnorm}[1]{\Vert{#1}\Vert}
\newtheorem{lemma}{Lemma}

\newtheorem{theorem}{Theorem}

\newcommand{\tr}[2]{\textnormal{Tr}_{#1}(#2)}

\newcommand{\abs}[1]{\left | {#1}\right |}

\newcommand\numeq[1]%
  {\stackrel{\scriptscriptstyle(\mkern-1.5mu#1\mkern-1.5mu)}{=}}
\newcommand\numleq[1]%
  {\stackrel{\scriptscriptstyle(\mkern-1.5mu#1\mkern-1.5mu)}{\leq}}
\newcommand\numgeq[1]%
  {\stackrel{\scriptscriptstyle(\mkern-1.5mu#1\mkern-1.5mu)}{\geq}}

\newcommand{\symbolinline}[2]{%
  \tikz[baseline=-0.5ex]{%
    \ifthenelse{\equal{#1}{circle}}{%
      \node[circle,draw,minimum size=8pt,inner sep=0pt] (c) {};}{}
    \ifthenelse{\equal{#1}{doublecircle}}{%
      \node[circle,draw,minimum size=8pt,inner sep=0pt] (c) {};
      \node[circle,draw,minimum size=4pt,inner sep=0pt] at (c) {};}{}
  }%
}
\usepackage{tikz}

\begin{document}
\title{Dissipative preparation of injective tensor network states}
\author{Drishti Baruah}
\email{drishti.baruah@mpq.mpg.de}
\affiliation{Max Planck Institute of Quantum Optics, Hans Kopfermann Str. 1, 85748 Garching, Germany}
\affiliation{Munich Center for Quantum Science and Technology, Schellingstr. 4, 80799 Munich, Germany}
\author{Georgios Styliaris}
\affiliation{Max Planck Institute of Quantum Optics, Hans Kopfermann Str. 1, 85748 Garching, Germany}
\affiliation{Munich Center for Quantum Science and Technology, Schellingstr. 4, 80799 Munich, Germany}
\author{J.~Ignacio Cirac}
\affiliation{Max Planck Institute of Quantum Optics, Hans Kopfermann Str. 1, 85748 Garching, Germany}
\affiliation{Munich Center for Quantum Science and Technology, Schellingstr. 4, 80799 Munich, Germany}
\author{Rahul Trivedi}
\email{rahul.trivedi@mpq.mpg.de}
\affiliation{Max Planck Institute of Quantum Optics, Hans Kopfermann Str. 1, 85748 Garching, Germany}
\affiliation{Munich Center for Quantum Science and Technology, Schellingstr. 4, 80799 Munich, Germany}

\begin{abstract}
The preparation of tensor network states is a fundamental prerequisite for a wide range of quantum simulation tasks. While many unitary protocols for preparing these states have been investigated, dissipative state preparation provides a powerful alternative since it can be robust to noise and initialization errors. In this paper, we construct both continuous-time and discrete-time geometrically local dissipative processes whose unique steady state is a given injective tensor network  state. Our method prepares all injective matrix product states on $N$ sites to an error $\varepsilon$ in $O(\log (N/\varepsilon))$ time, yielding an exponential improvement over previously known dissipative preparation schemes. For two and higher-dimensional tensor network states, we prove that when the tensors of the state are \emph{highly injective}, the constructed dissipative processes are rapid-mixing i.e., they prepare a state $\varepsilon$-close to the $N$-site target state in $O(
 \log (N/\varepsilon))$ time. For these states, our approach provides a polynomial speedup over known unitary methods for states defined on lattices and an exponential speedup for states on general bounded-degree graphs. We corroborate our theoretical results with numerical studies that indicate that the dissipative protocol can rapidly prepares states outside the high-injectivity assumption.
\end{abstract}
\maketitle
\emph{Introduction.---} Preparing many-body states on quantum devices is a prerequisite for a wide range of quantum algorithms. In many cases, the relevant target states are tensor network states, which form a physically motivated and broadly applicable class. For instance, tensor network states serve as highly correlated initial states for quantum simulation~\cite{bauer2023quantum}, trial wave functions for variational algorithms~\cite{cerezo2021variational}, states to encode information in error correction~\cite{kitaev2003fault}, and resourceful states for quantum metrology~\cite{giovannetti2006quantum}. Prominent examples are Matrix Product States (MPS) in one dimension, and their higher-dimensional generalizations, Projected Entangled Pair States (PEPS)~\cite{cirac_mps}. While the tensor network description provides an efficient specification of complex many-body states, circumventing the need to explicitly store exponentially many amplitudes, the tensors themselves do not directly correspond to unitary gates acting on physical degrees of freedom. Consequently, their efficient preparation needs to be addressed separately. Several closed-system approaches have been developed to prepare MPS and specific subclasses of PEPS, including quantum circuit constructions~\cite{schon2005sequential,wei2022sequential,malz2024}, protocols based on measurements and feedforward~\cite{sahay2025classifying,zhang2024characterizing,stephen2024preparing,ren2025efficient}, and adiabatic state preparation methods~\cite{Ge_2016}.
 
A different paradigm for preparing MPS and PEPS treats them as ground states of local and frustration-free Hamiltonians~\cite{fannes1992finitely,nachtergaele1996spectral,perez_garcia2008PEPS}. In this setting, state preparation can be thought of as a cooling process, which lowers the energy via engineered dissipation through interactions with an environment. This paradigm can be formulated as a Lindbladian, or a quantum circuit of channels, that drives the system toward a unique steady state corresponding to the target ground state~\cite{verstraete2009quantum,kraus2008preparation,roy2020measurement,zhou2021symmetry,cubitt2023dissipative}, or to a thermal state at a target temperature~\cite{brandao2019finite,chen2023quantum,rouze2024optimal,bergamaschi2025quantum}.  However, cooling techniques that rely solely on the Hamiltonian, without using explicit knowledge of the target ground state, are fundamentally limited beyond 1D due to complexity-theoretic barriers\footnote{Even for 2D geometrically local Hamiltonians that are commuting and frustration-free, with a unique (gapped) product ground state, cooling cannot be efficient in general; the required preparation time can scale exponentially with system size, since this setting already suffices to encode hard classical SAT problems.}. In contrast to purely unitary preparation models, in the dissipative paradigm the initial state is arbitrary, and the resulting protocols are therefore expected to be robust against noise and errors. This robustness can be rigorously established for local observable expectation values, e.g., when the underlying process is local and rapidly mixing~\cite{cubitt2015stability}. For general tensor network families, dissipative preparation schemes have been established for injective MPS~\cite{verstraete2009quantum} and for injective PEPS with a commuting parent Hamiltonian, satisfying an additional condition~\cite{verstraete2009quantum}, and for matrix product density operators families at the fixed point~\cite{liu2025parent}.

Here, we present a provably efficient dissipative algorithm for preparing MPS and PEPS with highly injective tensors. Our results are formulated for both the analog Lindbladian setting and the discrete quantum channel model, where the target tensor network state is the unique fixed point. In both formulations, the generators and channels are geometrically local, and the mixing time of the resulting dynamics scales logarithmically with system size. As such, the preparation is provably robust for local expectation values~\cite{cubitt2015stability}. We show that the required sufficient condition for our preparation method holds for arbitrary injective MPS, even when they are not initially highly injective, by blocking a finite number of consecutive sites. For PEPS, injectivity alone is not sufficient to guarantee logarithmic-time preparation~\cite{verstraete2006criticality}. However, we expect that a similar blocking argument can be applied to injective PEPS when the associated parent Hamiltonian is gapped.

For the MPS case, our algorithm yields an exponential speedup over the previously best known dissipative preparation scheme with mixing time $O(N^{\log N})$ for $N$ sites~\cite{verstraete2009quantum}, while also implying stability of local observables due to rapid-mixing~\cite{cubitt2015stability}. At the same time, it matches the asymptotically optimal preparation algorithm in the local circuit model, achieving logarithmic complexity in $N$~\cite{malz2024}. Moreover, it provides a polynomial speedup over the state-of-the-art method for highly injective PEPS based on adiabatic evolution~\cite{Ge_2016}, whose runtime scales as $O(\log^{d+1}(N/\varepsilon))$ in $d$ spatial dimensions with error $\varepsilon$. Finally, our approach is the first to achieve logarithmic-time preparation for highly injective PEPS on expander graphs.

\emph{Preliminaries.---}
Following the notation in~\cite{cirac_mps}, we first formally define a  projected entangled pair state (PEPS) on an undirected graph $\mathcal{G} = (\mathcal{V}, \text{E})$ with vertices $\mathcal{V}=\{1,2,\ldots,N\}$, together with a set of edges $\text{E}$ connecting the vertices. We will restrict ourselves to graphs $\mathcal{G}$ with bounded-degree $\mathfrak{d}_0$ i.e., the number of vertices connected to any one vertex by an edge $\leq \mathfrak{d}_0$. This is a natural setting in many-body physics and includes lattices in finite-number of dimensions as well as certain geometrically non-local models \cite{breuckmann2021ldpc, hastings2021ldpc}. As depicted in Fig.~\ref{fig:peps}, To specify a PEPS $\ket{\psi}$, we first place on each edge $e = (i, j) \in \text{E}$ an entangled state $\ket{\phi_0}_{e} \in \mathbb{C}^{D} \otimes \mathbb{C}^{D}$, where $D$ is the bond-dimension of the PEPS. Each vertex now holds several qudits, one from each edge incident on the vertex. Next, on the qudits at vertex $i$, we now apply a local linear map $A_i$
\begin{equation}\label{eq:tensorA}
A_i:  {(\mathbb{C}^{D})}^{\otimes \vert\mathcal{N}_i\vert} \longrightarrow \mathbb{C}^{d},
\end{equation} 
where $\mathcal{N}_i = \{j\in \mathcal{V} : (i, j) \in \text{E}\}$ is the neighbourhood of $i$ and $d$ is the physical dimension associated with each site. Mathematically, the resulting PEPS $\ket{\psi} \in {(\mathbb{C}^{d})}^{\otimes N}$ is given by 
\begin{equation}
\label{eq:PEPS-def}
\ket{\psi} = \left( \bigotimes_{i\in \mathcal{V}} A_i \right) \bigotimes_{e\in \text{E}} \vert\phi_0\rangle_{e}.
\end{equation}
The maps $\{A_i\}_{i \in \mathcal{V}}$ will also be referred to as the tensors specifying the state $\ket{\psi}$. Since each $A_i$ is local, it cannot increase the entanglement entropy of the original set of entangled pairs across any bipartition. Therefore, by construction, the state $\ket{\psi}$ will have an entanglement area law i.e., the entanglement entropy of any subset of sites $S$ will scale at-most as the graph-area of $S$ (i.e., the number of graph edges connecting $S$ to $S^c$).
\begin{figure}[t]
    \centering
    \includegraphics[width=1.0\linewidth]{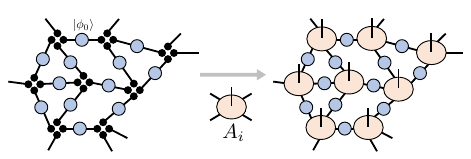}
    \caption{\justifying Projected entangled pair state (PEPS) on an arbitrary graph. Each vertex hosts a local tensor with one physical index and virtual indices associated with incident edges.}
    \label{fig:peps}
\end{figure}

We will restrict ourselves to a class of PEPS called \emph{injective} PEPS~\cite{perez_garcia2008PEPS}---the PEPS $\ket{\psi}$ is called injective if all the maps $A_i$ are \emph{injective} i.e. there exists a left inverse $A^{-1}_i$
\begin{equation}
A_i^{-1}: \mathbb{C}^{d} \longrightarrow 
{(\mathbb{C}^{D})}^{\otimes \vert\mathcal{N}_i\vert} 
\end{equation}
such that $A_i^{-1}A_i = I$. Injective tensor-network states on geometrically local graphs are believed to accurately describe ground-states of gapped non-topological Hamiltonians~\cite{schuch2010topo}. They are also studied as the ``typical" states within the space of tensor-network states since almost all tensors $A_i$ will be invertible under blocking~\cite{perez_garcia2022random}. In this paper, we will present dissipative processes that prepare injective PEPS as their unique fixed points---we will present both a continuous-time process, specified by a Lindbladian, as well as a discrete time-process, specified by a sequence of channels [Fig.~\ref{fig:peps}], which would be more amenable to implementation on digital platforms. Both of these processes will be geometrically local with respect to the underlying graph $\mathcal{G}$. Furthermore, we will provide evidence that these processes are rapid-mixing when the PEPS $\ket{\psi}$ have exponentially decaying correlations i.e., they prepare the tensor-network state in $O(\log(N))$ time.

\emph{Continuous-time dissipative protocol.---}
We first consider a setting relevant to analog simulators \cite{kashyap2025accuracy, zanardi_lindbladian_2016} and first construct a local Lindbladian $\mathcal{L}$ which has the target state $\ket{\psi}$ as its unique fixed point i.e., $\ket{\psi}$ is the only state that satisfies $\mathcal{L}(\ket{\psi}\!\bra{\psi})=0$.  We refer to $\mathcal{L}$ as the \emph{Parent Lindbladian} of the PEPS $\ket{\psi}$. Given the tensors $\{A_i\}_{i \in \mathcal{V}}$, we construct $\mathcal{L}$ as follows:
\begin{subequations}\label{eq:lindbladian_peps}
    \begin{align}
        \mathcal{L}&=\sum\limits_{e\in \text{E}}\mathcal{L}_{e} + \sum_{i\in \mathcal{V}}\mathcal{L}_i.
    \end{align}
    Here 
    \begin{enumerate}
        \item[(i)]Defining $S_i = \text{Im}(A_i)$ as the image of the map $A_i$ defined in Eq.~\eqref{eq:tensorA}, $\mathcal{L}_i$ is a Lindbladian acting on site $i$ given by
        \[
        \mathcal{L}_i = \frac{1}{\text{dim}(S_i)} P_{S_i} \text{Tr}_i(P_{S_i}^\perp (\cdot)) - \frac{1}{2}\{P_{S_i}^\perp, \cdot\}. 
        \]
        \item[(ii)] For every edge $e = (i, j) \in \text{E}$,
        \[
        \mathcal{L}_e = \sum_{\alpha = 1}^{D^2 - 1}\bigg(L_{\alpha, e}(\cdot) L_{\alpha, e} - \frac{1}{2}\{L_{\alpha, e}^\dagger L_{\alpha, e}, \cdot \}\bigg),
        \]
        is a two-site Lindbladian acting on sites $i, j$ with
            \begin{align}
        L_{\alpha,(i,j)}&=\vcenter{\hbox{\includegraphics[scale=0.5]{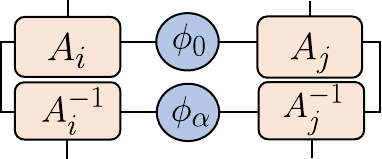}}}.
    \end{align}
    where $\{\ket{\phi_\alpha}\}_{\alpha = 1}^{D^2 - 1}$ being a set of orthonormal states also orthogonal to $\ket{\phi_0}$ used in constructing $\ket{\psi}$. 
    \end{enumerate}
\end{subequations}
Physically, the single-site Lindbladian $\mathcal{L}_i$ ensures that the single-site reduced state is in the correct subspace of the physical Hilbert space. Furthermore, when the jump operators of  $\mathcal{L}_{e}$ for $e = (i, j)$ act on a state that is not the target PEPS $\ket{\psi}$, they first undo the maps $A_i$ and $A_j$ on each edge $(i, j)$, then project the state shared by the edge $(i,j)$ onto the target state $|\phi_0\rangle_{(i,j)}$ and finally re-apply the tensors. Our first result provides a characterization of the spectrum of the parent Lindbladian.
\begin{theorem}\label{theorem:theorem_fp}
    The PEPS $\ket{\psi}$ $[$Eq.~\eqref{eq:PEPS-def}$]$ is unique fixed-point of Lindbladian $\mathcal{L}$ defined in Eq.~\eqref{eq:lindbladian_peps}. Furthermore, $\mathcal{L}$ has no other purely imaginary eigenvalues.
\end{theorem}
\noindent Thus, starting from any initial state, we are guaranteed to eventually evolve to the target PEPS $\ket{\psi}$ under the Lindbladian $\mathcal{L}$. Theorem \ref{theorem:theorem_fp} indicates that it is not possible for the parent Lindbladian to have multiple steady-states, or have oscillations (due to purely imaginary eigenvalues) persisting at long-times. From a practical standpoint, this theorem rigorously proves that the proposed dissipative protocol is robust to initialization errors unlike unitary protocols which tend to be sensitive to the initial state. The proof of this theorem utilizes the fact that the jump operators $L_{\alpha, e}$ defined in Eq.~\eqref{eq:lindbladian_peps} are commuting and nilpotent, and is provided in the supplementary material.

Next, we analyze how the mixing time of $\mathcal{L}$ i.e., long does the parent Lindbladian take to prepare $\ket{\psi}$. While it is generally hard to provide mixing time-bounds on many-body Lindbladians, we show that when the tensors $A_i$ are close to isometric, the parent Lindbladian is rapid mixing~\cite{lucia2015rapid_mixing,cubitt2015stability}. We introduce a parameter $\delta$ controlling how non-isometric the tensors $\{A_i\}_{i \in \mathcal{V}}$ are: More precisely,
\begin{equation}\label{eq:isometry}
    \delta = \min_{i \in \mathcal{V}} \norm{A_i^\dagger A_i - I}.
\end{equation} 
For small $\delta$, the  PEPS is known to have exponentially decaying correlation functions~\cite{schuch2011}, which is not true for injective PEPS in general~\cite{verstraete2006criticality, malz2025complexity}. 
Our next result shows that when $\delta \leq \delta_\text{th}$, with $\delta_\text{th}$ being controlled entirely by the underlying graph $\mathcal{G}$, $\mathcal{L}$ rapidly reaches the target PEPS.

\begin{theorem}
    \label{theorem:peps}
    There exists a constant $\delta_\textnormal{th}$ dependent only on the maximum degree $\mathfrak{d}_0$ of the graph $\mathcal{G}$ such that if $\delta \leq \delta_\textnormal{th}$, then\footnote{The $1\to 1$ norm of a superoperator $\mathcal{E}$ is the induced trace norm, defined as $\norm{\mathcal{E}}_{1\to 1} = \sup_{\norm{X}_1 = 1} \norm{\mathcal{E}(X)}_1$, where $\norm{\cdot}_1$ denotes the Schatten 1-norm (trace norm) for operators.}
    $$\norm{e^{\mathcal{L}t} - \textnormal{Tr}(\cdot) \ket{\psi}\!\bra{\psi}}_{1\to 1}\leq N\mathfrak{d}_0 \exp(-\Theta(t)).$$ 
     Furthermore, the threshold $\delta_\textnormal{th} = \mathcal{O}(1/\mathfrak{d}_0)$ as $\mathfrak{d}_0 \to \infty$.
\end{theorem}
\noindent Thus, to prepare the state $\ket{\psi}$ to a trace-norm error $\varepsilon$, would require time $O(\log(N \mathfrak{d}_0/\varepsilon))$ as long as $\delta \leq \delta_\text{th}$. 

For injective matrix product states (MPS), which correspond to injective PEPS with $\mathcal{G}$ being chosen as a 1D lattice, the constraint $\delta \leq \delta_\text{th}$ can always be satisfied by coarse-graining or blocking the MPS by a system-size independent block-size (i.e., treating multiple adjacent sites as a single site) \cite{verstraete2005rg}---we detail this in the supplementary material. Thus, we obtain the following result for preparing injective MPSs.
\begin{theorem}
For every injective matrix product state (MPS), there is a geometrically local 1D Lindbladian which has it as its unique steady state. Furthermore, this Lindbladian prepares the injective MPS to trace-norm error $\varepsilon$ in time  $O(\log(N/\varepsilon))$.   
\label{theorem:mpslogn}
\end{theorem}
\noindent This result is an exponential improvement over the best known dissipative method for preparing these states \cite{verstraete2009quantum}. Furthermore, although unitary state-preparation methods for injective MPS also achieve a similar scaling as the proposed dissipative method \cite{malz2024}, the dissipative protocol has the benefit of being robust to initialization errors as well as, since it is rapid-mixing, to errors in the configured Lindbladian \cite{cubitt2015stability}.

In 2 or higher dimensions, the proposed Lindbladian only is guaranteed to be rapid-mixing for small enough $\delta$. Remarkably, for graphs of constant degree, when $\delta$ is small enough, the parent Lindbladian always prepares the state in $O(\log N)$ time, where $N$ is the system-size, \emph{irrespective of the graph}. This is in stark contrast to the best known unitary methods which are based on adiabatic algorithms: In particular, for geometrically local PEPS with small $\delta$, the best known unitary method would yield a preparation time scaling as $O(\log^{d + 1}N)$ in $d-$dimensions and as $O(N^3)$ on geometrically non-local graphs \cite{Ge_2016}. The dissipative method, therefore, yields a polynomial speedup over the best-known unitary method for highly injective PEPS defined on lattices in two and higher dimensions, and an exponential speedup over the best-known unitary method for highly injective PEPS on geometrically non-local, but constant degree, graphs.

Finally, we remark that the constraint on $\delta$ in theorem \ref{theorem:peps} is only a \emph{sufficient} condition for rapid-mixing of the parent Lindbladian: later on, we provide numerical examples showing that the parent Lindbladian mixes in $\log N$ time even when the analytical condition in theorem \ref{theorem:peps} is violated.


The proof of theorem \ref{theorem:peps} is detailed in the Supplementary Material: The key idea behind the proof is it analyze the time-evolution of the \emph{parent Hamiltonian} of the injective PEPS. The parent Hamiltonian $H$ of an injective PEPS is defined via \cite{cirac_mps}
\begin{subequations}
\label{eq:parent_ham}
    \begin{align}
        H&=\sum\limits_{(i,j)\in \text{E}}\ h_{i,j},\\
        h_{i,j}&= \vcenter{\hbox{\includegraphics[scale=0.5]{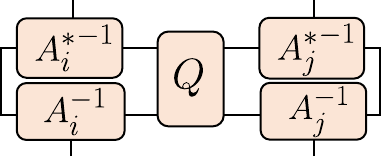}}},
    \end{align}
\end{subequations}
where $Q = I - \ket{\phi_0}\!\bra{\phi_0}$. For injective PEPS, the parent Hamiltonian is known to be frustration free with the PEPS as its unique ground state with ground-state energy $0$. Furthermore, for sufficiently small $\delta$, $H$ can be shown to have a non-zero gap $\Delta_H = \Omega(1)$ between its first excited state and ground state, that does not vanish with system-size \cite{schuch2011}. Due to this energy gap, for any state $\rho$, it follows that $\langle \psi | \rho \ket{\psi} \geq 1 - \text{Tr}(\rho H)/\Delta_H$. Therefore, if we can show that with evolution under the parent Lindbladian, the expectation value of $H$ decreases exponentially with time, it follows that the distance between the evolved state and $\ket{\psi}$ also decreases exponentially with $t$.

We then analyze the time-evolution of the parent Hamiltonian expected value: $\langle H(t) \rangle =\tr{}{H\rho(t)}= \sum_{e} \langle h_{e}(t) \rangle$. In the Heisenberg picture,
\[
\frac{d}{dt}\langle h_e(t)\rangle = \sum_{e' \in \text{E}} \text{Tr}[\mathcal{L}_{e'}^\dagger(h_{e}) \rho(t)],
\]
with $\rho(t) = e^{\mathcal{L}t}(\rho(0))$ for an arbitrary initial state $\rho(0)$. Now, $\mathcal{L}^\dagger_{e'}(h_e)$ is zero if $e'$ and $e$ do not share a vertex. Furthermore, in the supplementary material, we show that $\text{Tr}[\mathcal{L}_e^\dagger(h_e)\rho(t)] \leq -(1 - O(\delta))\langle h_e(t)\rangle$ i.e., for any edge $e$, the corresponding jump operators decrease the local energy corresponding to that edge if $\delta$ is sufficiently small. Finally, the jump operators corresponding to an edge $e' \neq e$ (that intersects with $e$) can possibly increase the local energy corresponding to $e$: Nevertheless, we show in the supplemental material that $\text{Tr}[\mathcal{L}_{e'}^\dagger(h_e)\rho(t)] \leq O(\delta)\langle h_{e'}(t)\rangle$. Therefore, for sufficiently small $\delta$, the total energy still decreases at a constant rate, which combined with together with the gap of the parent Hamiltonian implies theorem \ref{theorem:peps}. Finally, since the number of edges whose jumps can increase the energy corresponding to any one edge grows with the degree $\mathfrak{d}_0$ of the graph $\mathcal{G}$, we need to ensure that $\delta < O(1/\mathfrak{d}_0)$ for the energy to exponentially decrease in time.


\emph{Discrete-time dissipative protocol.---}
Depending on the PEPS, the Lindbladian constructed above could require complicated jump operators which may not be easily realizable in experimental platforms. As an alternative, we consider how to implement the dissipative protocol in discrete time, enabling execution on digital quantum hardware and compatibility with quantum error correction. A natural strategy is to approximate the Lindblad evolution with a sequence of quantum channels is to use a digital Lindbladian simulation algorithm~\cite{cleve_lindblad, kliesch2011dissipation,xiehang25}. However, due to the irreversible nature of dissipative dynamics, all currently available digital simulation methods for geometrically-local Lindblad evolution require circuit depths that scaling as $\text{poly}(N, t)$, where $N$ is the number of qudits and $t$ is the Lindblad-evolution time. Therefore, even though the Lindbladian constructed above mixes in $\sim \log(N)$ time, its digital simulation could only be guaranteed to mix in $\text{poly}(N)$ time.

Instead, we directly construct a discrete-time quantum channel that shares the same fixed point and mixing properties as the target Lindbladian. We define the local channels as follows:

\begin{subequations}\label{eq:discrete_channel_peps}
    \begin{enumerate}
        \item[(i)] Defining $S_i = \text{Im}(A_i)$ as the image of the map $A_i$ defined in Eq.~\eqref{eq:tensorA}, we construct a single-site channel $\mathcal{E}_i$ for every vertex $i \in \mathcal{V}$ that acts as the identity on $S_i$ and redistributes population from the orthogonal complement into $S_i$:
        \[
            \mathcal{E}_i(\cdot) = P_{S_i} (\cdot) P_{S_i} + \frac{P_{S_i}}{\text{dim}(S_i)} \text{Tr}(P_{S_i}^\perp (\cdot)).
        \]
        
        \item[(ii)] For every edge $e = (i,j) \in \text{E}$, we construct a two-site channel $\mathcal{E}_e$ with Kraus operators $K_{0,e}, K_{1,e} \dots K_{D^2-1,e}$ where:
        \begin{align}\label{eq:kraus_ops}
            K_{\alpha,e} &= \begin{cases} 
            \sqrt{\Gamma} L_{\alpha,e} & \text{for } \alpha \neq 0, \\ 
            (I - 2\Gamma H_{\text{eff},e})^{1/2} & \text{for } \alpha = 0, \end{cases}
        \end{align}
        with $L_{1,e}, L_{2,e} \dots L_{D^2-1,e}$ being the jump operators from Eq.~(4) and $H_{\text{eff},e} = \frac{1}{2}\sum_\alpha L_{\alpha,e}^\dagger L_{\alpha,e}$. We remark that $\Gamma$ is chosen to be small enough such that $\Gamma H_{\text{eff},e} \preceq I$ and thus $K_{0,e}$ is well-defined.
    \end{enumerate}
\end{subequations}
Now, to construct the full channel on all the sites, we first partition the edges $\text{E}$ into disjoint subsets $\text{E}_1, \text{E}_2 \dots \text{E}_k$ such that any two edges in the same subset do not share a vertex. In 1D for example, we have two such subsets; one consisting of odd pairs $\{ (i,i+1): i \textnormal{ is odd}\}$ and the other consisting of even pairs $\{ (i,i+1): i \textnormal{ is even}\}$. More generally, this can be done explicitly for any bounded-degree graph using a graph coloring algorithm~\cite{vizing_theorem} and it can be guaranteed that the number of such subsets $k \leq O(\frak{d}_0)$. For each subset $\text{E}_i$, we construct the channel $\mathcal{E}_{\text{E}_i}$ via
\[
\mathcal{E}_{\text{E}_i} = \prod_{e \in \text{E}_i} \mathcal{E}_e.
\]
We note that since no two edges in $\text{E}_i$ share the same vertex, the two-site channels $\mathcal{E}_e$ for $e \in \text{E}_i$ commute and thus can be applied simultaneously. Now, we construct the channel $\mathcal{E}$ by picking one of the channels $\mathcal{E}_{\text{E}_i}$ uniformly at random and applying them. Mathematically, this corresponds to a channel $\mathcal{E}$ given by
\begin{align}\label{eq:local_channel}
\mathcal{E} = \left (\frac{1}{k} \sum_{i = 1}^k \mathcal{E}_{\text{E}_i}\right ) \circ\prod_{i\in \mathcal{V}}\mathcal{E}_i 
\end{align}
\begin{figure*}[htpb]
    \includegraphics[width=1.0\linewidth]{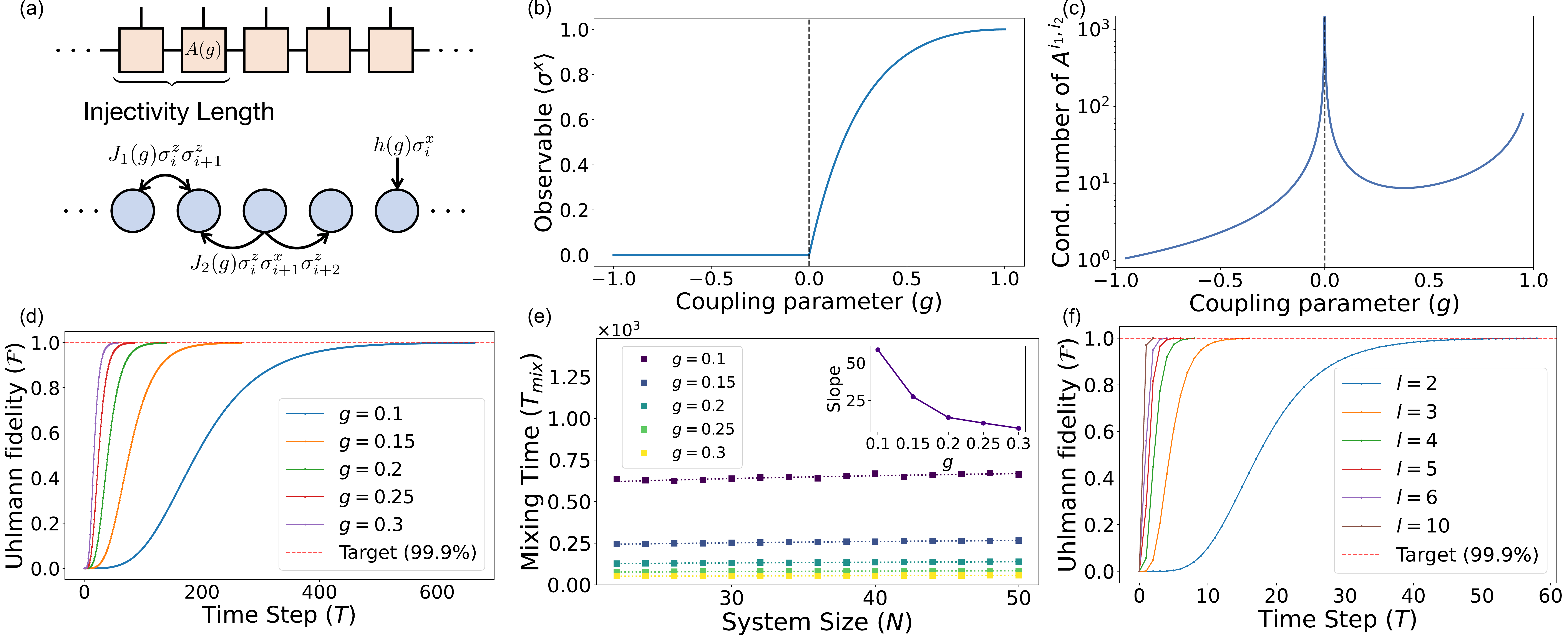}

    \caption{\justifying Numerical study of the dissipative protocol beyond the high-injectivity assumption. (a) A schematic illustrating the local 1D chain and the Parent Hamiltonian interactions. (b) Magnetization $\langle \sigma^x \rangle$ as a function of the coupling $g$. The sharp discontinuity at $g=g_c=0$ signals a quantum phase transition into a GHZ-like critical state. (c)  Condition number of the two-site blocked tensor $A^{i_1,i_2}$ as a function of the coupling $g$, diverging at the phase transition point $g=0$ where the tensor becomes non-injective. (d) Fidelity convergence for a fixed block size $l=2$ across different couplings $g$ ($N=50, l=2$). (e) Logarithmic scaling of $T_\textnormal{mix}$ with system size $N$ across different couplings $g$ ($l=2)$. The inset displays the extracted slope of $t_\textnormal{mix}$ vs $\log N$ across $g$. The dotted lines show the fits to $\log N$. (f) Fidelity convergence for a fixed coupling $g=0.3$ across different block sizes $l$ ($N=60$).}
    \label{fig:combined_numerics}
\end{figure*}

For this channel, we can prove a result analogous to theorem \ref{theorem:theorem_fp} which shows that its repeated application eventually always prepares the target PEPS.
\begin{theorem}\label{theorem:discrete_fp}
    The PEPS $\ket{\psi}$ $[$Eq.~\eqref{eq:PEPS-def}$]$ is the unique fixed-point of the quantum channel $\mathcal{E}$ in Eq.~\eqref{eq:local_channel}. Furthermore, $\mathcal{E}$ has no other eigenvalues with magnitude equal to 1.
\end{theorem}
\noindent Furthermore, analogous to theorem \ref{theorem:peps}, we can also rigorously characterize the mixing time of the channel $\mathcal{E}$ when $\delta$ [Eq.~\eqref{eq:isometry}] is sufficiently small i.e., smaller than a threshold $\delta_\text{th}$ which depends entirely on the degree of the graph $\mathcal{G}$.
\begin{theorem}
\label{theorem:discrete}
  There exists a constant $\delta_\textnormal{th}$ dependent only on the maximum degree $\mathfrak{d}_0$ of the graph $\mathcal{G}$ such that if $\delta \leq \delta_\textnormal{th}$, then choosing $\Gamma = \Theta(\delta)$ independent of the system-size $N$,
\[\norm{\mathcal{E}^T - \textnormal{Tr}(\cdot) \ket{\psi}\!\bra{\psi}}_{1 \to 1} \leq O(N \mathfrak{d}_0 e^{-\Theta(T)}).
\]
Furthermore, the threshold $\delta_\textnormal{th} = \mathcal{O}(1/\mathfrak{d}_0)$ as $\mathfrak{d}_0 \to \infty$.
\end{theorem}

\noindent Thus, the channel $\mathcal{E}$ prepares the state $\ket{\psi}$ to an error $\varepsilon$ after $O(\log(N\frak{d}_0/\varepsilon)) $ applications as long as $\delta \leq \delta_\text{th}$. We remark here that while the discrete-time dissipative dynamics has the same steady-state as that of the continuous time dynamics constructed earlier, their dynamics would in general be different. In particular, for the channel $\mathcal{E}$ to well approximate the Lindbladian, we would in general need to choose $\Gamma \sim 1/\text{poly}(N)$, while in theorem \ref{theorem:discrete} it is chosen to be independent of $N$. The proof of this theorem, detailed in the supplementary material, follows a strategy similar to that in the case of theorem \ref{theorem:peps}: We track the energy of the parent Hamiltonian during the dissipative dynamics and show that it continuously decreases with repeated applications of $\mathcal{E}$ if $\Gamma$ is appropriately chosen.

As an illustration, we consider dissipatively preparing a one-parameter family of translationally invariant MPS~\cite{wolf2006quantum} with periodic boundary conditions, defined by the local tensors [Fig.~\ref{fig:combined_numerics}(a)]:
\begin{equation}
    A^0 = \frac{1}{\sqrt{1 + g}}\begin{pmatrix} 1 & g \\ 0 & 0 \end{pmatrix}, \quad A^1 = \frac{1}{\sqrt{1 + g}} \begin{pmatrix} 0 & 0 \\ 1 & 1 \end{pmatrix},
\end{equation}
where $g \in (-1, 1)$. This MPS is the ground state of the $g-$dependent Hamiltonian:
\begin{equation}
    \label{eq:parent_ham_mps_fam}
    H(g)=\sum_i \big(J_1(g) \sigma^z_i\sigma^z_{i+1} +J_2(g)\sigma^z_i\sigma^x_{i+1}\sigma^z_{i+2}- h(g) \sigma^x_i\big),
\end{equation}
where $J_1(g) = 2(g^2-1), h(g) = (1+g)^2 $ and $J_2(g) = (g-1)^2$. At $g = 0$, this Hamiltonian exhibits a second-order phase transition [Fig.~\ref{fig:combined_numerics}(b)] which can be seen as a non-analyticity in $O(g) = \langle \sigma^x \rangle$. Apart from $g = 0$, this MPS is injective after blocking two sites i.e., two-site tensor $A^{i_1, i_2} = A^{i_1} A^{i_2}$ is injective. This is shown in Fig.~\ref{fig:combined_numerics}(c) where the minimum singular value of the two site tensor $A^{i_1, i_2}$ viewed as a map from the bond-space to the physical-space is calculated as a function of $g$: At $g = 0$, this singular value becomes 0 corresponding to the tensor becoming non-injective otherwise it is non-zero.

We initialize the qubits in the maximally mixed state and compute $\rho(T) = \mathcal{E}^T(\rho(0))$ and track its Uhlmann Fidelity $\mathcal{F} = \langle \psi | \rho(T) \ket{\psi}$ relative to the target MPS $\ket{\psi}$
as a function of $T$. Fig.~\ref{fig:combined_numerics}(d) shows this fidelity as a function of $T$ for different $g$---we observe that the fidelity exponentially converges to $1$, and the convergence time increases as $g \to 0$ i.e., we go closer to the phase-transitin point. Fig.~\ref{fig:combined_numerics}(e) shows the time $T_\text{mix}$ needed for the fidelity to reach 99.9$\%$ as a function of $N$: We see a logarithmic scaling of this time with $N$ as suggested by our theoretical results (Theorem \ref{theorem:discrete}). As physically expected, the slope of $T_\text{mix}$ vs $\log N$ also increases as $g$ is tuned closer to the phase transition point. Finally, in Fig.~\ref{fig:combined_numerics}(f), we simulate the discrete-time protocol with $l$ tensors blocked into one tensor: We observe that increasing $l$ reduces the mixing time of the dissipative process. This can be attributed to the fact that larger $l$ makes the tensor closer to isometric and thus reducing its mixing time.

Finally, we remark that the dissipative method works well without blocking sites (i.e, for $l=2$) [Fig.~\ref{fig:combined_numerics}(d - e)] even reasonably close to the phase transition point where the tensor is expected to be highly non-isometric (for example at $g=0.1$, where the condition number is $\sim20$  [Fig.~\ref{fig:combined_numerics}(c)]). This provides evidence that the dissipative processes constructed in this paper can be used beyond the regime where we are currently able to theoretically prove rapid-mixing.

\emph{Conclusion and outlook.---} We have presented a method for the efficient dissipative preparation of injective tensor network states, including matrix product states (MPS) and $\delta$-isometric projected entangled pair states (PEPS). This is achieved via constructing both continuous-time and discrete-time processes whose unique steady state is the target PEPS. These processes are found to be rapid-mixing, preparing the target state in time logarithmic in system size, while simultaneously guaranteeing robustness to local errors. In 1D, our method achieves an exponential speedup over previously known dissipative MPS protocols, matching the optimal scaling of unitary methods, and in higher dimensions, it provides a polynomial speedup over known unitary methods for PEPS with geometrically local graphs and an exponential speedup for those with non-geometrically local graphs.

Future theoretical directions include extending these constructions to non-injective PEPS and implementing dissipative state preparation with error-corrected digital quantum devices. On the experimental side, our protocols could be implemented in a range of platforms, including superconducting qubits with engineered dissipation, cold atoms in optical lattices where local interactions and entanglement can be tuned, and solid-state systems supporting highly controllable interactions. These dissipative methods open a pathway toward fast and noise-resilient preparation of complex many-body states for quantum simulation, computation, and metrology.

\emph{Acknowledgements.---}  R.T. acknowledges funding from the European Union’s Horizon Europe research and innovation program under grant agreement number 101221560 (ToNQS). RT and GS acknowledge funding from the Deutsche Forschungsgemeinschaft (DFG, German Research Foundation) under Germany’s Excellence Strategy – EXC-2111 – 390814868. J.I.C. acknowledges the project THEQUCO within the Munich Quantum Valley (MQV), which is supported by the Bavarian state government with funds from the Hightech Agenda Bayern Plus. He also acknowledges support from the German Federal Ministry of Education and Research (BMBF) through the funded project AL-MANAQC, grant number 13N17236 within the research program ``Quantum Systems”.
\bibliography{paper_bib}

\onecolumngrid
\newpage
\begin{center}
\textbf{\large Supplemental material to \\ ``Dissipative preparation of injective tensor network states''}
\vspace{0.3cm}

Drishti Baruah,$^{1, 2, *}$ Georgios Styliaris,$^{1, 2}$ J.~Ignacio Cirac,$^{1, 2}$ and Rahul Trivedi$^{1, 2, \dagger}$

\vspace{0.1cm}
{\small \textit{$^1$Max Planck Institute of Quantum Optics, Hans Kopfermann Str. 1, 85748 Garching, Germany}}\\
{\small \textit{$^2$Munich Center for Quantum Science and Technology, Schellingstr. 4, 80799 Munich, Germany}}\\
\vspace{0.1cm}
{\small (Dated: \today)}
\end{center}

\begingroup
\renewcommand{\thefootnote}{\fnsymbol{footnote}}
\footnotetext[1]{%
\begin{tabular}[t]{@{}l@{}}
drishti.baruah@mpq.mpg.de \\
\llap{$^\dagger$\,}rahul.trivedi@mpq.mpg.de
\end{tabular}}
\endgroup
\vspace{0.3cm}
This Supplemental Material is organized as follows: First, in \cref{sec:notation}, we set up a few key definitions and notational setup that we use throughout the text. In \cref{sec:cont_time}, we include the proofs of Theorems 1 and 2. In particular, we go in detail through the continuous-time dissipative protocol, constructing the parent Lindbladian which prepares the target PEPS. We prove the uniqueness of its steady state and rapid mixing for highly injective PEPS. In \cref{sec:MPS}, we show how blocking (which is coarse-graining) sites for injective matrix product states can satisfy the high-injectivity assumption, implying rapid mixing for any injective MPS, proving Theorem 3. Finally, in \cref{sec:discrete_time}, we include the proofs of Theorems 4 and 5. We go through the discrete-time dissipative protocol, constructing a quantum channel to prepare the target PEPS and similarly as in the continuous-time case, we prove uniqueness of steady state and rapid mixing for highly-injective PEPS.
\vspace{0.3cm}
\section{Preliminaries and Notations}\label{sec:notation}
We begin by establishing the mathematical definitions and notations used throughout.
\subsection{Norms and Distances}
For bounded linear operators acting on a Hilbert space $\mathcal{H}$, we utilize the following norms:
\begin{itemize}
    \item \textbf{Operator Norm:} Denoted by $\Vert X \Vert$, representing the maximum singular value of $X$.
    \item \textbf{Trace Norm (Schatten 1-norm):} Denoted by $\Vert X \Vert_1 = \tr{}{\sqrt{X^\dagger X}}$. The distance between two quantum states $\rho$ and $\sigma$ is naturally bounded using the trace distance $\frac{1}{2}\Vert \rho - \sigma \Vert_1$.
    \item \textbf{Induced Superoperator Norm ($1\rightarrow 1$):} For a superoperator (such as a quantum channel $\mathcal{E}$ or Liouvillian $\mathcal{L}$), we use the induced trace norm defined as $\Vert \mathcal{E} \Vert_{1\rightarrow 1} = \sup_{\Vert X \Vert_1 = 1} \Vert \mathcal{E}(X) \Vert_1$.
\end{itemize}

\subsection{Tensor Network Setup}
We consider Projected Entangled Pair States (PEPS) defined on a graph $\mathcal{G}=(\mathcal{V},E)$ with maximum degree $\mathfrak{d}_{0}$. The construction assumesentangled states $\ket{\phi_0} \in \mathbb{C}^{D} \otimes \mathbb{C}^{D}$ on each edge $e \in E$, mapped to the physical space by local linear maps (tensors) $A_i: (\mathbb{C}^{D})^{\otimes |\mathcal{N}_i|} \to \mathbb{C}^d$ at each vertex $i \in \mathcal{V}$. 

The global target state is defined via the tensor network contraction mapping the virtual bond indices to the physical indices:
\begin{equation}\label{eq:PEPS-def}
    \ket{\psi} = \left(\bigotimes_{i\in\mathcal{V}} A_i\right) \bigotimes_{e\in E} \ket{\phi_0}_e.
\end{equation}

We always assume the maps $A_i$ are injective. We denote $S_i := \textnormal{Im}(A_i) \subseteq \mathbb{C}^d$. Also we will assume the convention that $A_i^{-1} P_{S_i}^\perp = 0$ i.e. $A_i^{-1}$ is the pseudo-inverse of $A_i$.

We quantify  how far the tensors are from being singular via a parameter $\delta \le 1/2$, such that $\Vert A_i^\dagger A_i - I \Vert \le \delta$. For the uniqueness of the steady state, we do not make any assumptions about $\delta$. For the proofs analysing the mixing time, we require $\delta$ to be less than a certain threshold which we derive later. We often refer to the injective PEPS that also satisfy $\delta \leq \delta_\textnormal{th}$ as \textit{highly injective}.

Next we formulate the parent Hamiltonian $H$ which has the injective PEPS $\psi$ [Eq.~\eqref{eq:PEPS-def}] as its unique ground state:\begin{subequations}
\label{eq:parent_ham}
    \begin{align}
        H&=\sum\limits_{(i,j)\in \text{E}}\ h_{i,j},\\
        h_{i,j}&= \vcenter{\hbox{\includegraphics[scale=0.5]{parent_ham.pdf}}},
    \end{align}
\end{subequations}
where $Q = I - \ket{\phi_0}\!\bra{\phi_0}$. The Hamiltonian is constructed such that it projects any state orthogonal to the target shared state onto the target state itself. We define the ground-state subspace corresponding to each edge $(i,j)$ as
\begin{equation}
\label{eq:projector}
    S_{i,j}=\{\ket{\psi}: h_{i,j}\ket{\psi}=0\}.
\end{equation}
Note that $\ket{\psi}\in S_{i,j} \Rightarrow \exists \ket{X}: \ket{\psi}=\vcenter{\hbox{\includegraphics[scale=0.5]{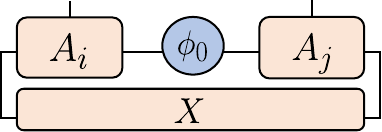}}}.$
\subsection{List of Notations}
The following is a summary of the operators, subspaces, and dynamical maps utilized throughout the main text and supplementary proofs.
\begin{itemize}
    \item[$S_i$:] The physical subspace at site $i$, defined as the image of the local tensor: $S_i := \textnormal{Im}(A_i) \subseteq \mathbb{C}^d$.
    \item[$P_{S_i}$:] The orthogonal projector onto the valid site subspace $S_i$.
    \item[$F_i$:] The single-site violation operator tracking leakage out of the valid tensor network manifold, defined as the projection onto the orthogonal complement: $F_i := P_{S_i}^\perp = I - P_{S_i}$.
    \item[$h_e$:] The local, positive semi-definite Parent Hamiltonian term corresponding to edge $e=(i,j)$. It is constructed to perfectly annihilate the target state ($h_e \ket{\psi} = 0$).
    \item[$S_e$:] The target ground-state subspace shared between two adjacent sites connected by edge $e$, defined algebraically as $S_e = \{\ket{\phi} : h_e \ket{\phi} = 0\}$.
    \item[$P_{S_e}$:] The orthogonal projector onto the edge ground-state subspace $S_e$. Its orthogonal complement is $P_{S_e}^\perp$.
    \item[$L_{\alpha, e}$:] The local dissipative jump operators acting on edge $e$. For $\alpha > 0$, they strictly annihilate the target state ($L_{\alpha, e} \ket{\psi} = 0$).
    \item[$H_{\textnormal{eff},e}$:] The effective Hamiltonian governing the diagonal non-Hermitian back-action of the jump operators on edge $e$, explicitly defined as $H_{\textnormal{eff},e} = \frac{1}{2} \sum_{\alpha>0} L_{\alpha,e}^\dagger L_{\alpha,e}$.
    \item[$\mathcal{L}_i$:] The local single-site correction Lindbladian designed to dissipate population from $S_i^\perp$ into $S_i$.
    \item[$\mathcal{L}_e$:] The local two-site Parent Lindbladian driving the system into the edge ground-state subspace $S_e$.
    \item[$\mathcal{E}_i, \mathcal{E}_e$:] The corresponding single-site and two-site discrete-time quantum channels (CPTP maps) utilized to approximate the continuous-time Lindbladian evolution.
\end{itemize}
\section{Analysis of continuous-time processes (Proof of theorems 1 and 2)}\label{sec:cont_time}
\subsection{Uniqueness of steady state}

To establish theorem 1, we first introduce two lemmas which allow us to characterize more general Lindbladians which share structural properties to the parent Lindbladian.
\begin{lemma}\label{lemma:unique_fp_lindblad} 
    Consider a Lindbladian $\mathcal{L}$ defined on the Hilbert space $\mathcal{H} = (\mathbb{C}^{d + 1})^{\otimes N}$ given by
    \[
    \mathcal{L} = \sum_{i = 1}^N \sum_{\alpha  =1}^d \mathcal{D}_{L_{\alpha, i}} \text{ where } L_{\alpha, i} = X \ket{0}_i\!\bra{\alpha} X^{-1},
    \]
    where $X:\mathcal{H} \to \mathcal{H}$ is an invertible operator. Then, $\mathcal{L}(\rho_\infty) = 0$ if and only if $\rho_\infty \propto  \ket{\psi}\!\bra{\psi}$ where $\ket{\psi} = X \ket{0}^{\otimes N}$.
\end{lemma}
\begin{proof}
It will also be convenient to note that the jump operators $L_{\alpha, e}$ are commuting i.e. $[L_{\alpha, i}, L_{\alpha', i'}] = 0$. It is easy to check that the jump operators $L_{\alpha, i}$ satisfy $L_{\alpha, i} \ket{\psi} = 0 \ \forall \alpha, i$: therefore, $\ket{\psi}$ is indeed a fixed point of $\mathcal{L}$.

Next, we show that $\ket{\psi}$ is the unique fixed point. Assume that $\rho_\infty$ is a fixed point of the parent Lindbladian: We then obtain that
\begin{align}\label{eq:fixed_point}
    \sum_{\alpha, i} L_{\alpha, i} \rho_\infty L_{\alpha, i}^\dagger = \frac{1}{2}\sum_{\alpha, i}\{L_{\alpha, i}^\dagger L_{\alpha, i}, \rho_\infty \}.
\end{align}
Suppose $\mathcal{S}_\infty = \text{span}\{\ket{e} : \ket{e} \textnormal{ are eigenvectors of }\rho_\infty\}$ and $P_{\mathcal{S}_\infty}$ be the projector on $\mathcal{S}_\infty$ and $P^\perp_{\mathcal{S}_\infty} = I -P_{\mathcal{S}_\infty}$. Then, since $\rho_\infty P_{\mathcal{S}_\infty}^\perp = P_{\mathcal{S}_\infty}^\perp \rho_\infty = 0$, it follows that
\begin{align}
    P^\perp_{\mathcal{S}_\infty}\bigg(\sum_{\alpha, i} L_{\alpha, i} \rho_\infty L_{\alpha, i}^\dagger\bigg) P_{\mathcal{S}_\infty}^\perp = 0,
\end{align}
from which it follows that $P_{\mathcal{S}_\infty}^\perp L_{\alpha, i} P_{\mathcal{S}_\infty} = 0 \ \forall \alpha, e$ or, equivalently,
\begin{align}\label{eq:invariant_support}
\forall \ket{\phi} \in \mathcal{S}_\infty: L_{\alpha, i} \ket{\phi} \in \mathcal{S}_\infty.
\end{align}
In other words, $\mathcal{S}_\infty$ is an invariant subspace of the jump operators $L_{\alpha, i}$. Furthermore, using $L_{\alpha, i} \ket{\psi} = 0$, it also follows from Eq.~\eqref{eq:fixed_point} that
\begin{align}
    \bra{\psi}\bigg(\sum_{\alpha, i}L_{\alpha, i} \rho_\infty L_{\alpha, i}^\dagger \bigg) \ket{\psi} = 0,
\end{align}
from which it follows that that $\bra{\psi}L_{\alpha, i} P_{\mathcal{S}_\infty} = 0$ or equivalently
\begin{align}\label{eq:orthogonality_under_jump}
\forall \ket{\phi} \in \mathcal{S}_\infty : \bra{\psi} L_{\alpha, i}\ket{\phi} = 0.
\end{align}
We further note that Eqs.~\eqref{eq:invariant_support} and~\eqref{eq:orthogonality_under_jump} imply that $\forall \ket{\phi} \in \mathcal{S}_\infty, i_1, i_2 \dots i_m \in \{1, 2 \dots N\}, \alpha_1, \alpha_2 \dots\alpha_m \in \{1, 2 \dots d\}$,
\begin{align}\label{eq:orthogonality_multiple_jumps}
 \bra{\psi} L_{\alpha_1, i_1}L_{\alpha_2, i_2} \dots L_{\alpha_m, i_m} \ket{\phi} = 0.
\end{align}
This can be shown easily by noting that, from Eq.~\eqref{eq:invariant_support}, $\ket{\phi'} = L_{\alpha_2, i_2} L_{\alpha_3, i_3} \dots L_{\alpha_m, i_m} \ket{\phi} \in \mathcal{S}_\infty$ and then from Eq.~\eqref{eq:orthogonality_under_jump}, $\langle \psi| L_{\alpha_1, i_1} \ket{\phi'} = 0$.

We will next show that $\mathcal{S}_\infty \subseteq \text{span}(\ket{\psi})$, which would imply that $\rho_\infty \propto \ket{\psi}\!\bra{\psi}$ i.e., the only possible fixed point of the parent Lindbladian $\mathcal{L}$ is $\ket{\psi}$. Suppose $\ket{\phi} \in \mathcal{S}_\infty$: We first express it as
\[
\ket{\phi} = X \sum_{\alpha_1, \alpha_2 \dots \alpha_N = 0}^{d} \phi_{\alpha_1, \alpha_2 \dots \alpha_N} \ket{\alpha_1, \alpha_2 \dots \alpha_n},
\]
which is always possible since $X$ is invertible and hence $\{X\ket{\alpha_1, \alpha_2 \dots \alpha_m}\}_{\alpha_1, \alpha_2 \dots \alpha_m \in \{1, 2 \dots d\}}$ is a linearly independent basis. Now, for a given $\alpha_1, \alpha_2 \dots \alpha_N \neq 0, 0\dots 0$, it follows that
\begin{align*}
\bigg(\prod_{i: \alpha_i \neq 0} L_{\alpha_i, i} \bigg) \ket{\phi} = \phi_{\alpha_1, \alpha_2 \dots \alpha_N} \ket{\psi}.
\end{align*}
However, from Eq.~\eqref{eq:orthogonality_multiple_jumps}, it then follows that
\begin{align*}
    \phi_{\alpha_1, \alpha_2 \dots \alpha_N} = \frac{1}{\norm{\ket{\psi}}^2} \bra{\psi}\bigg(\prod_{i: \alpha_i \neq 0} L_{\alpha_i, i} \bigg) \ket{\phi} = 0,
\end{align*}
and thus $\phi_{\alpha_1, \alpha_2 \dots \alpha_N} = 0$ for all $\alpha_1, \alpha_2 \dots \alpha_N \neq 0, 0 \dots 0$. However, this implies that $\ket{\phi} = \phi_{0, 0 \dots 0} X \ket{0, 0 \dots 0}$ thus proving the lemma.
\end{proof}
\begin{lemma}\label{lemma:no_imaginary}
    Suppose $\mathcal{E}$ is a channel on a finite-dimensional Hilbert space $\mathcal{H}$ with a pure-state $\ket{\psi}$ as its unique fixed-point i.e., $\mathcal{E}(\rho_\infty) = \rho_\infty \implies \rho_\infty \propto \ket{\psi}\!\bra{\psi}$. Then, $\textnormal{spec}(\mathcal{E}) \subseteq \{z\in \mathbb{C}: |z| < 1\} \cup \{1\}$ i.e., $\mathcal{E}$ has no eigenvector other than $\ket{\psi}\!\bra{\psi}$ on its peripheral spectrum.
\end{lemma}
\begin{proof}
    Suppose $\mathcal{E} = \sum_\alpha K_\alpha (\cdot) K_\alpha^\dagger$ with $\sum_\alpha K_\alpha^\dagger K_\alpha = I$. Since $\mathcal{E}(\ket{\psi}\!\bra{\psi}) = \ket{\psi}\!\bra{\psi}$, we obtain that $K_\alpha \ket{\psi} = c_\alpha \ket{\psi}$ for some $c_\alpha \in \mathbb{C}$ satisfying $\sum_\alpha \abs{c_\alpha}^2 = 1$. Decomposing $\mathcal{H} = \mathcal{K} \oplus \mathcal{K}^\perp$ where $\mathcal{K} = \text{span}(\ket{\psi})$, the Kraus operators $\{K_\alpha\}_\alpha$ can thus be expressed as
    \[
    K_\alpha = \begin{bmatrix}
        c_\alpha & \bra{v_\alpha} \\
        0 & A_\alpha
    \end{bmatrix},
    \]
    for some $\ket{v_\alpha} \in \mathcal{K}^\perp, A_\alpha : \mathcal{K}^\perp \to \mathcal{K}^\perp$. Note that it follows from $\sum_\alpha K_\alpha^\dagger K_\alpha = I$ that
    \begin{align}\label{eq:normalization_blocks}
    \sum_\alpha A_\alpha^\dagger A_\alpha + \sum_\alpha \ket{v_\alpha}\!\bra{ v_\alpha} = I.
    \end{align}
    
    Consider now the completely-positive (CP) map on the Hilbert space $\mathcal{K}^\perp$, $\mathcal{F} = \sum_\alpha A_\alpha (\cdot) A_\alpha^\dagger$: we first show that its spectral radius is strictly smaller than 1. Since from Eq.~\eqref{eq:normalization_blocks} we obtain $\sum_\alpha A_\alpha^\dagger A_\alpha \preceq I$, it also follows that $\mathcal{F}$ has a spectral radius $\leq 1$. Suppose that $\mathcal{F}$ had a spectral-radius = 1: since $\mathcal{F}$ is CP, we obtain that $\exists \sigma \succeq 0$ such that $\mathcal{F}(\sigma) = \sigma$. Then, 
    \[
    \mathcal{E}\bigg(\begin{bmatrix}
        0 & 0 \\
        0 & \sigma
    \end{bmatrix}\bigg) = \begin{bmatrix}
        \sum_\alpha \bra{v_\alpha} \sigma \ket{v_\alpha} & \sum_\alpha \bra{v_\alpha}\sigma A_\alpha^\dagger \\
        \sum_\alpha A_\alpha \sigma \ket{v_\alpha} & \sigma
    \end{bmatrix}.
    \]
    However, since $\mathcal{E}$ is also trace-preserving, we obtain that $\sum_\alpha \bra{v_\alpha} \sigma \ket{v_\alpha} = 0$. Since $\sigma \succeq 0$, we obtain that $\bra{v_\alpha} \sigma \ket{v_\alpha} = 0 \ \forall \alpha$ and thus $\sigma \ket{v_\alpha} = \bra{v_\alpha} \sigma = 0 \ \forall \alpha$. Thus, we conclude that $\mathcal{E}$ has another fixed-point which contradicts the assumption that $\ket{\psi}\!\bra{\psi}$ is its only fixed point. Therefore, we conclude that the spectral radius of $\mathcal{K}$ is strictly smaller than 1.

     Now, assume that $\exists X$ such that $\mathcal{E}(X) = \lambda X$ for some $\lambda \neq 1$ with $\abs{\lambda}= 1$. We will assume that
     \begin{align}\label{eq:form_X}
     X = \begin{bmatrix}
         a & \bra{x} \\
         \ket{y} & W
     \end{bmatrix},
     \end{align}
     for some $a \in \mathbb{C}, \ket{x}, \ket{y} \in \mathcal{K}^\perp, W:\mathcal{K}^\perp \to \mathcal{K}^\perp$. It is useful to note that
     \begin{align}\label{eq:channel_expression}
     \mathcal{E}(X) = \begin{bmatrix}
         a + \sum_\alpha (c_\alpha^* \bra{v_\alpha} y\rangle + c_\alpha \bra{x} v_\alpha \rangle + \bra{v_\alpha} W \ket{v_\alpha} )  & \sum_\alpha(c_\alpha \bra{x} A_\alpha^\dagger + \bra{v_\alpha} W A_\alpha^\dagger) \\
         \sum_\alpha(c_\alpha^* A_\alpha \ket{y} + A_\alpha W \ket{v_\alpha}) & \mathcal{F}(W)
     \end{bmatrix}
     \end{align}
    Thus, it follows from $\mathcal{E}(X) = \lambda X$ that $\mathcal{F}(W) = \lambda W$. However, by assumption $\abs{\lambda} = 1$ and, as shown above, the spectral radius of $\mathcal{F}$ is strictly less than 1. Therefore, $W = 0$: using this together with Eq.~\eqref{eq:channel_expression} and $\mathcal{E}(X) = \lambda X$, we then obtain that
    \begin{align}\label{eq:cross_diagonal_terms}
    \sum_\alpha c_\alpha^* A_\alpha \ket{x} =   \lambda^* \ket{x} \text{ and }\sum_\alpha c_\alpha^* A_\alpha \ket{y} = \lambda \ket{y}.
     \end{align}
     Now,
     \[
     \abs{\lambda}^2\norm{\ket{x}}^2 = \left \Vert{\sum_\alpha c_\alpha^* A_\alpha \ket{x}} \right \Vert^2 \numleq{1} \bigg(\sum_\alpha \abs{c_\alpha}^2\bigg) \bigg(\sum_\alpha \norm{A_\alpha \ket{x}}^2\bigg) \numleq{2} \norm{\ket{x}}^2,
     \]
     where in (1) we have used the Cauchy-Schwarz inequality and in (2) we have used the fact that $\sum_\alpha \abs{c_\alpha}^2 = 1$ and $\sum_\alpha A_\alpha^\dagger A_\alpha \preceq I$. Thus, since $\abs{\lambda}^2 = 1$, both the inequalities in (1) and (2) have to be equalities---this implies that
     \begin{subequations}\label{eq:cauchy_eq}
     \begin{align}
         &\sum_\alpha \norm{A_\alpha \ket{x}}^2 = \sum_\alpha \bra{x} A_\alpha^\dagger A_\alpha \ket{x} =  \norm{\ket{x}}^2, \\
         &\exists \ket{f} \in \mathcal{K}^\perp \text{ such that } \forall \alpha : A_\alpha \ket{x} = c_\alpha \ket{f}.
     \end{align}
     \end{subequations}
     From Eqs.~\eqref{eq:cross_diagonal_terms} and (\ref{eq:cauchy_eq}b), it thus follows that $\lambda^* \ket{x} = \sum_\alpha c_\alpha^* A_\alpha \ket{x} = \ket{f}$ and therefore 
     \begin{align}\label{eq:cauchy_equality_1}
     \forall \alpha: A_\alpha \ket{x} = c_\alpha \lambda^* \ket{x}.
     \end{align}
     Furthermore, from Eqs.~(\ref{eq:cauchy_eq}a) and \eqref{eq:normalization_blocks}, it follows that $\sum_{\alpha}\abs{ \bra{x}v_\alpha\rangle}^2 = 0$ and therefore
     \begin{align}\label{eq:cauchy_equality_2}
         \forall \alpha: \bra{v_\alpha}x\rangle = 0.
     \end{align}
     Equations \eqref{eq:cauchy_equality_1} and \eqref{eq:cauchy_equality_2} together imply that
     \[
     \forall \alpha : K_\alpha \begin{bmatrix}
         0 \\
         \ket{x}
     \end{bmatrix} = c_\alpha \lambda^* \begin{bmatrix}
         0 \\
         \ket{x}
     \end{bmatrix} \implies \mathcal{E}\bigg(\begin{bmatrix}
         0 & 0 \\
         0 & \ket{x}\!\bra{x}
     \end{bmatrix}\bigg) = \abs{\lambda}^2 \begin{bmatrix}
         0 & 0 \\
         0 & \ket{x}\!\bra{x}
     \end{bmatrix}.
     \]
     Since $\abs{\lambda} = 1$ by assumption, if $\ket{x} \neq 0$, we would have constructed a fixed-point of $\mathcal{E}$ other than $\ket{\psi}$, contradicting with its uniqueness. Thus, we conclude that $\ket{x} = 0$. A similar argument can be made starting from Eq.~\eqref{eq:cross_diagonal_terms} to also conclude that $\ket{y} = 0$. Returning to Eq.~\eqref{eq:form_X} and Eq.~\eqref{eq:channel_expression}, we would then obtain that $\mathcal{E}(X) = \lambda X \implies a = \lambda a$: if $\lambda \neq 1$, then this would imply $a = 0$. Thus, we conclude that if $\mathcal{E}(X) = \lambda X$ for some $\lambda \neq 1$ with $\abs{\lambda} = 1$ then $X = 0$ which proves the lemma.
\end{proof}
Now we look at the parent Lindbladian of interest and its properties. To recap, we denote the image of the map $A_i$ defined in Eq.~\eqref{eq:tensorA} as $S_i = \text{Im}(A_i)$. We define the two-site Lindblad jump operators as: \begin{align}
        L_{\alpha,(i,j)}&=\vcenter{\hbox{\includegraphics[scale=0.5]{Lindblad_op.pdf}}}.
    \end{align}
    where $\{\ket{\phi_\alpha}\}_{\alpha = 1}^{D^2 - 1}$ being a set of orthonormal states also orthogonal to $\ket{\phi_0}$ used in constructing $\ket{\psi}$. 
The parent Lindbladian $\mathcal{L}$ is then given by:
\begin{equation}
    \mathcal{L} = \sum_{e\in E}\mathcal{L}_{e}+\sum_{i\in\mathcal{V}}\mathcal{L}_{i},
\end{equation}
where the single-site and two-site Lindbladians are respectively defined as:
\begin{align}
    \mathcal{L}_i(\cdot) &= \frac{1}{\dim(S_i)} P_{S_i} \textnormal{Tr}_i(P_{S_i}^\perp (\cdot)) - \frac{1}{2}\{P_{S_i}^\perp, \cdot\}, \label{eq:single_site}\\
    \mathcal{L}_e(\cdot) &= \sum_{\alpha=1}^{D^2-1} \left( L_{\alpha,e} (\cdot) L_{\alpha,e}^\dagger - \frac{1}{2}\{L_{\alpha,e}^\dagger L_{\alpha,e}, \cdot\} \right).
\end{align} Now we prove that $\mathcal{L}$ has the injective PEPS $\psi$ as its unique fixed point and it also has no purely imaginary eigenvalues.
\begin{proof}[Proof of Theorem 1]
The proof proceeds in three main steps: restricting the dynamics to the valid physical subspace, establishing the uniqueness of the fixed point within that subspace, and characterizing the peripheral spectrum.
 
We denote $\rho_\infty$ as any state satisfying $\mathcal{L}(\rho_\infty) = 0$. Uisng  \eqref{eq:single_site}, we obtain that:\begin{align}
    \text{Tr}\left(P_{S_i}^\perp \mathcal{L}_i(\rho_\infty)\right) &= \text{Tr}\left( P_{S_i}^\perp \left[ \frac{1}{\dim(S_i)} P_{S_i} \text{Tr}_i(P_{S_i}^\perp \rho_\infty) \right] \right) - \frac{1}{2} \text{Tr}\left( P_{S_i}^\perp \{P_{S_i}^\perp, \rho_\infty\} \right) \notag\\
    &\overset{(1)}{=} 0 - \frac{1}{2} \text{Tr}\left( P_{S_i}^\perp P_{S_i}^\perp \rho_\infty + P_{S_i}^\perp \rho_\infty P_{S_i}^\perp \right) \notag\\
    &\overset{(2)}{=} -\text{Tr}(P_{S_i}^\perp \rho_\infty),
\end{align}
where in (1) we used the orthogonality $P_{S_i}^\perp P_{S_i} = 0$, and in (2) we used $(P_{S_i}^\perp)^2 = P_{S_i}^\perp$ alongside the cyclic property of the trace.

Next, we must verify that no other terms in the full Lindbladian $\mathcal{L} = \sum_{e\in E}\mathcal{L}_{e}+\sum_{i\in\mathcal{V}}\mathcal{L}_{i}$ contribute to this trace. 
For any site $j \neq i$, the generator $\mathcal{L}_j$ acts on a disjoint local Hilbert space. Because Lindbladians are trace-preserving ($\text{Tr}_j \circ \mathcal{L}_j = 0$ where $\text{Tr}_j$ denotes the partial trace), it follows immediately that $\text{Tr}(P_{S_i}^\perp \mathcal{L}_j(\rho_\infty)) = 0$. 
Furthermore, since the jump operators $L_{\alpha, e}$ for any incident edge $e$ are supported strictly within $S_i$, they satisfy $P_{S_i}^\perp L_{\alpha, e} = L_{\alpha, e}^\dagger P_{S_i}^\perp = 0$. This orthogonality immediately yields $\text{Tr}(P_{S_i}^\perp \mathcal{L}_e(\rho_\infty)) = 0$. Summing up, we get
\begin{equation}
    \text{Tr}\left(P_{S_i}^\perp \mathcal{L}(\rho_\infty)\right) =  \text{Tr}\left(P_{S_i}^\perp \mathcal{L}_i(\rho_\infty)\right) = -\text{Tr}(P_{S_i}^\perp \rho_\infty) = 0.
\end{equation}
Because $\rho_\infty$ is a positive semi-definite density matrix, $\text{Tr}(P_{S_i}^\perp \rho_\infty) = 0$ implies that $P_{S_i}^\perp \rho_\infty P_{S_i}^\perp = 0$. Thus, the steady state $\rho_\infty$ has strictly zero support in the orthogonal complement $S_i^\perp$ for all $i \in \mathcal{V}$.

Within the valid subspace, the steady-state condition reduces to $\sum_{e \in E}\mathcal{L}_e(\rho_\infty) = 0$. Next we wish to invoke Lemma~\ref{lemma:unique_fp_lindblad} and for that we set $\bigotimes_{e\in E} |\phi_0\rangle_e$ as the local vacuum state $|0\rangle^{\otimes N}$ from the Lemma, and let the global invertible operator $X = \bigotimes_{i \in \mathcal{V}} A_i$, then the two-site jump operators take the explicit form:
\begin{equation}
    L_{\alpha, (i,j)} = X \ket{\phi_0}_{(i,j)}\! \bra{\phi_\alpha} X^{-1}.
\end{equation}
Then Lemma~\ref{lemma:unique_fp_lindblad} implies $\mathcal{L}(\rho_\infty) = 0$ if and only if $\rho_\infty \propto  \ket{\psi}\!\bra{\psi}$ where $\ket{\psi} = X \ket{0}^{\otimes N}=\left( \bigotimes_{i\in \mathcal{V}} A_i \right) \bigotimes_{e\in \text{E}} \vert\phi_0\rangle_{e}.$

Finally, we must ensure that $\mathcal{L}$ has no purely imaginary eigenvalues. To see this, note that if $\mathcal{L}$ had purely imaginary eigenvalues then $\mathcal{E}=e^{\mathcal{L}t}$, for $t >0$, would have an eigenvalue on the unit circle. From the uniqueness established above, we already have that $\mathcal{E}(\rho_\infty) = \rho_\infty$ implies $\rho_\infty \propto \ket{\psi}\!\bra{\psi}$. Then Lemma 2 guarantees that $\mathcal{E}$ has no other eigenvalues on the unit circle. Thus, we conclude that $\ket{\psi}\!\bra{\psi}$ is the unique steady state of the full Lindbladian $\mathcal{L}$.
\end{proof}

\subsection{Rapid mixing of the parent Lindbladian for near-isometric PEPS}

\subsubsection{Preliminaries}

\begin{lemma}[Local Parent Hamiltonian Bounds]
\label{lemma:local_ham_bounds}
For any edge $e=(i,j)$, let $h_{i,j}$ be the local parent Hamiltonian term and $\Tilde{h}_{i,j}$ be the projector term composed of isometric tensors. Under the $\delta$-isometry condition with $\delta \le 1/2$, the following bounds hold:
\begin{enumerate}
    \item The deviation from the ideal projector is bounded by $\smallnorm{h_{i,j} - \tilde{h}_{i,j}} \le 8\delta$.
    \item The local spectral gap of $h_{i,j}$ is lower-bounded by $\Delta_{h_{i,j}} \ge 1 - 8\delta$, where the spectral gap is defined to be the difference in energy between its first excited state and ground state.
    \item The operator norm is bounded by $\left\Vert h_{i,j} \right\Vert \le 4$.
\end{enumerate}
\end{lemma}
\begin{proof}
Suppose $A_i$ has a singular value decomposition $A_i=U_i \Sigma_i V_i^\dagger$ for unitaries $U_i, V_i$. The $\delta$-isometry condition implies $\norm{A_i^\dagger A_i - I}=\left\Vert \Sigma_i^2 - I \right\Vert \le \delta$, which yields $$\left\Vert \Sigma_i - I \right\Vert \leq\max \left ( 1 - \sqrt{1-\delta}, \sqrt{1+\delta} -1 \right )\leq \max \left (\frac{\delta }{1+ \sqrt{1-\delta}}, \frac{\delta}{1+\sqrt{1+\delta}}\right )\le \delta.$$ Let $\Tilde{h}_{i,j}$ be defined by replacing the local tensors with perfect isometries:
\begin{equation*}
    \Tilde{h}_{i,j}=\vcenter{\hbox{\includegraphics[scale=0.5]{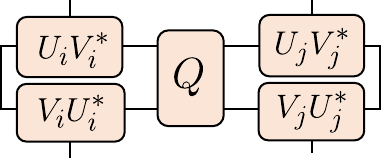}}}.
\end{equation*}
By telescoping the difference $h_{i,j} - \Tilde{h}_{i,j}$ and using the inequality $\left\Vert A_i^{-1} \right\Vert \le (1-\delta)^{-1/2} \le \sqrt{2}$, we obtain:
\begin{align}
    \smallnorm{ h_{i,j}-\Tilde{h}_{i,j} } &\le \left[(1-\delta)^{-\frac32}+(1-\delta)^{-1}+(1-\delta)^{-\frac12} +1\right]\delta \le 8\delta.
\end{align}
Note that $\Tilde{h}_{i,j}^2=\Tilde{h}_{i,j}$ and thus the spectrum of $\Tilde{h}_{i,j} \in \{0,1\}$ and the ground state energy of $h_{i,j}$ is $0$, Weyl's theorem guarantees the gap of $h_{i,j}$ satisfies:
\begin{equation}\label{eq:gap}
    \Delta_{h_{i,j}} \ge 1- 8 \delta.
\end{equation}
Furthermore, bounding the norm directly using $\left\Vert Q \right\Vert = 1$ yields:
\begin{equation}\label{eq:hij_norm}
    \left\Vert h_{i,j} \right\Vert \le \smallnorm{ A_i^{-1}}^2 \smallnorm{ A_j^{-1}}^2 \left\Vert Q \right\Vert \le \left(\frac{1}{1-\delta}\right)^2 \le 4.
\end{equation}
\end{proof}


Next we prove that the parent Hamiltonian possesses a non-vanishing spectral gap that is uniform in the system size.
\begin{lemma}[Uniform Spectral Gap]
\label{lemma:spectral_gap}
For $\delta < \frac{1}{4\mathfrak{d}_0 - 3}$, the Parent Hamiltonian $H$ exhibits a strictly positive spectral gap 
$$\Delta_H \geq \frac{1}{1+\delta}-\frac{4\delta(\mathfrak{d}_0-1)}{1-\delta^2} =\Omega(1) \text{ as }\mathfrak{d}_0 \to \infty,$$
where $\mathfrak{d}_0$ is the maximum degree of the graph.
\end{lemma}
\begin{proof}
We lower-bound the gap $\Delta_H$ by showing $H^2 - \Delta_H H \succeq 0$ via the martingale method~\cite{nachtergaele1996spectral}. 

The parent Hamiltonian terms can be written as $h_{i,j} = (A_i\otimes A_j)^{-\dagger} Q_{i,j} (A_i\otimes A_j)^{-1}$, where $Q_e$ are projectors acting on the virtual bonds. Using the $\delta-$isometry condition, we get $\smallnorm{A_i^\dagger A_i}\leq 1+\delta$, which implies $\smallnorm{(A_i^\dagger A_i)^{-1}}\geq 1/(1+\delta)$, which finally leads to 
\begin{equation}\label{eq:hsquare_norm}
    h_e^2 \succeq \frac{1}{1+\delta} h_e
\end{equation}
Expanding $H^2$ yields $\sum_e h_e^2 + \sum_{e_1 \neq e_2} h_{e_1} h_{e_2}$. Using Eq.~\eqref{eq:hsquare_norm}, we obtain:
\begin{equation}\label{eq:sum_ham_terms}
    H^2 - \Delta_H H \succeq \left( \frac{1}{1+\delta} - \Delta_H \right) \sum_e h_e + \sum_{e_1 \neq e_2} h_{e_1}h_{e_2}.
\end{equation}
If $e_1, e_2$ do not overlap, then $h_{e_1}$ and $h_{e_2}$ commute which together with the fact that they are both PSD implies $h_{e_1} h_{e_2} \succeq 0$. For overlapping edges $e_1 = (a,i)$ and $e_2 = (i,j)$, we evaluate the anti-commutator by mapping it to the virtual space.  Even though $Q_{e_1}$ and $Q_{e_2}$ share the tensor at site $i$, they act on different virtual bonds and therefore perfectly commute ($[Q_{e_1}, Q_{e_2}] = 0$).

Defining $X_i = A_i^{-1} A_i^{-\dagger}$ and $A=A_a \otimes A_i\otimes A_j$, the anti-commutator in the physical space is:
\begin{equation}
    h_{e_1} h_{e_2} + h_{e_2} h_{e_1} = A^{-\dagger} \left( Q_{e_1} X_i Q_{e_2} + Q_{e_2} X_i Q_{e_1} \right) A^{-1}.
\end{equation}
Now let $\mu = 1/(1-\delta)$ and $\nu = 1/(1+\delta)$. The $\delta$-isometry condition ensures $\nu I \preceq X_i \preceq \mu I$. Defining the positive semi-definite operator $Y_i = X_i - \nu I \succeq 0$ and substituting $X_i = Y_i + \nu I$ yields:
\begin{equation}
    Q_{e_1} X_i Q_{e_2} + Q_{e_2} X_i Q_{e_1} = 2\nu Q_{e_1} Q_{e_2} + Q_{e_1} Y_i Q_{e_2} + Q_{e_2} Y_i Q_{e_1}.
\end{equation}
For any positive semi-definite $Y_i \succeq 0$ and Hermitian $Q_{e_1}, Q_{e_2}$, the square $(Q_{e_1} + Q_{e_2}) Y_i (Q_{e_1} + Q_{e_2}) \succeq 0$. Expanding this and rearranging gives a strict lower bound for the cross terms:
\begin{equation}
    Q_{e_1} Y_i Q_{e_2} + Q_{e_2} Y_i Q_{e_1} \succeq - (Q_{e_1} Y_i Q_{e_1} + Q_{e_2} Y_i Q_{e_2}).
\end{equation}
Using the upper bound $Y_i \preceq (\mu - \nu)I$ and the projector property $Q_e^2 = Q_e$, we bound $-Q_{e_1} Y_i Q_{e_1} \succeq -(\mu - \nu)Q_{e_1}$. Therefore:
\begin{equation}
    Q_{e_1} X_i Q_{e_2} + Q_{e_2} X_i Q_{e_1} \succeq 2\nu Q_{e_1} Q_{e_2} - (\mu - \nu)(Q_{e_1} + Q_{e_2}).
\end{equation}
Since $2\nu Q_{e_1} Q_{e_2} \succeq 0$, we can drop it. Mapping back to the physical space by multiplying left and right by $A^{-\dagger}$ and $A^{-1}$ recovers the operator inequality:
\begin{equation}
    \label{eq:martingale}
    h_{e_1} h_{e_2} + h_{e_2} h_{e_1} \succeq -(\mu - \nu)(h_{e_1} + h_{e_2}).
\end{equation}
Let $z$ be the maximum number of overlapping edges for any single edge. Note that $z \le 2(\mathfrak{d}_0-1)$.
Now Eqs.~\eqref{eq:sum_ham_terms} and~\eqref{eq:martingale} together imply
\begin{equation}
    H^2 - \Delta_H H \succeq \left( \frac{1}{1+\delta} - \Delta_H \right) \sum_e h_e - z (\mu - \nu) \sum_e h_e.
\end{equation}
This directly means we satisfy $H^2 - \Delta_H H \succeq 0$ if:
\begin{equation}
    \mu - \nu \le \frac{1}{z} \left( \frac{1}{1+\delta} - \Delta_H \right).
\end{equation}
Solving for $\Delta_H$ and substituting $\mu - \nu = \frac{2\delta}{1-\delta^2}$ yields:
\begin{equation}
    \Delta_H \ge \frac{1}{1+\delta} - z \left( \frac{2\delta}{1-\delta^2} \right).
\end{equation}
To ensure the uniform gap is strictly positive ($\Delta_H > 0$), we require:
\begin{equation}
    \frac{1-\delta - 2\delta z}{1-\delta^2} > 0 \implies \delta < \frac{1}{2z + 1}.
\end{equation}
Since an edge overlaps with at most $z \le 2(\mathfrak{d}_0-1)$ other edges, guaranteeing $\delta < \frac{1}{4\mathfrak{d}_0 - 3}$ ensures the spectral gap $\Delta_H$ remains strictly positive and independent of $N$.
\end{proof}

We now state a technical lemma.
\begin{lemma}
    If we consider the neighboring subspaces $S_{i,j}$ and $S_{a,i}$, then \begin{align*}
        P_{S^\perp_{i,j}}P_{S_{a,i}}= P_{S^\perp_{i,j}} A_{S_{i,j}, S_{a,i}} P_{S^\perp_{i,j}}+ P_{S^\perp_{i,j}} B_{S_{i,j}, S_{a,i}} P_{S^\perp_{a,i}},
    \end{align*}
    where \begin{align*}
        A_{S_{i,j}, S_{a,i}}&=\sum\limits_{n=0}^\infty P_{S_{a,i}}\left((P_{S_{i,j}}P_{S_{a,i}})^n  - P_{S_{i,j}\cap S_{a,i}}\right)\\
        B_{S_{i,j}, S_{a,i}}&=\sum\limits_{n=0}^\infty \left((P_{S_{a,i}}P_{S_{i,j}})^n - P_{S_{i,j}\cap S_{a,i}}\right),
    \end{align*}
    and for $S_{i,j}$ and $S_{a,i}$ as defined in the text, \begin{equation*}
        \norm{A_{S_{i,j}, S_{a,i}}}, \norm{B_{S_{i,j}, S_{a,i}}}\leq \frac{1}{(1-138\delta)}.
    \end{equation*}\label{lemma:proj}
\end{lemma}
\begin{proof} We note that
    \begin{align}P_{S^\perp_{i,j}}P_{S_{a,i}} &=P_{S^\perp_{i,j}}P_{S_{a,i}}P_{S^\perp_{i,j}}+P_{S^\perp_{i,j}}P_{S_{a,i}}P_{S_{i,j}}\notag\\
        &=P_{S^\perp_{i,j}}P_{S_{a,i}}P_{S^\perp_{i,j}}+P_{S^\perp_{i,j}}P_{S_{a,i}}P_{S_{i,j}}P_{S^\perp_{a,i}} +P_{S^\perp_{i,j}}P_{S_{a,i}}P_{S_{i,j}}P_{S_{a,i}}\notag\\
        &=P_{S^\perp_{i,j}}\left \{\sum\limits_{n=0}^\infty P_{S_{a,i}}(P_{S_{i,j}}P_{S_{a,i}})^n\right \}P_{S^\perp_{i,j}} + P_{S^\perp_{i,j}}\left \{\sum\limits_{n=0}^\infty (P_{S_{a,i}}P_{S_{i,j}})^n\right \}P_{S^\perp_{a,i}}\notag\\
         &\overset{(1)}{=} P_{S^\perp_{i,j}} A_{S_{i,j}, S_{a,i}} P_{S^\perp_{i,j}} + P_{S^\perp_{i,j}} B_{S_{i,j}, S_{a,i}} P_{S^\perp_{a,i}},
        \end{align}
        where (1) holds because by definition of the projector $P_{S_{i,j}\cap S_{a,i}}$, we get $P_{S^\perp_{a,i}}P_{S_{i,j}\cap S_{a,i}}=P_{S^\perp_{i,j}}P_{S_{i,j}\cap S_{a,i}}=0.$ Also $ A_{S_{i,j}, S_{a,i}},  B_{S_{i,j}, S_{a,i}}$ are as defined in the lemma statement.

By alternating projectors' lemma~\cite{bhatia1997matrix}, we have 
\begin{align}
    \norm{(P_{S_{i,j}}P_{S_{a,i}})^n  - P_{S_{i,j}\cap S_{a,i}}},\norm{(P_{S_{a,i}}P_{S_{i,j}})^n - P_{S_{i,j}\cap S_{a,i}}} \leq (\cos{\theta})^n,
\end{align}
where \begin{align}
    \cos{\theta}=\max\limits_{v,u}\frac{\langle v,u\rangle}{\norm{v}\norm{u}},  v\in S_{i,j}\cap( S_{i,j} \cap  S_{a,i})^\perp,  u\in S_{a,i}\cap( S_{i,j} \cap  S_{a,i})^\perp.\label{eq:cos1}
\end{align}
We next bound $\cos{\theta}$. Note that since  $v \in S_{i,j}$ and $u \in S_{a,i}$, $\exists \ket{X_{i,j}}, \ket{X_{a,i}}$ such that
\begin{align}
    \ket{v}=  \vcenter{\hbox{\includegraphics[scale=0.5]{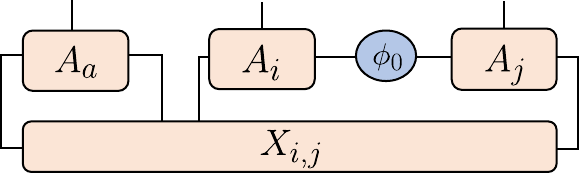}}},\ket{u}= \vcenter{\hbox{\includegraphics[scale=0.5]{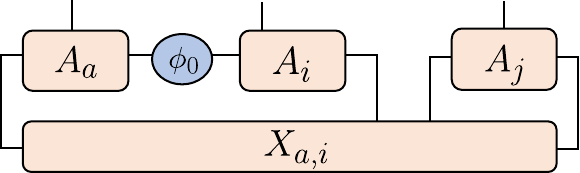}}}.
\end{align}
Then \begin{equation}
    \braket{v}{u}=\vcenter{\hbox{\includegraphics[scale=0.5]{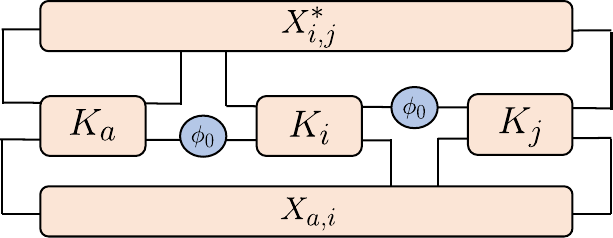}}}.
\end{equation}
Also note that since $v,u \in (S_{i,j} \cap S_{a,i})^\perp$, 
\begin{align}
    \vcenter{\hbox{\includegraphics[scale=0.5]{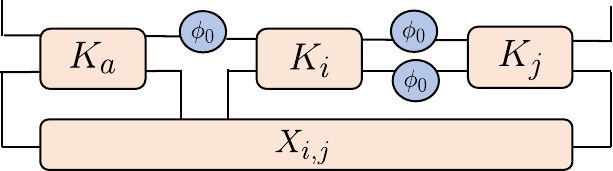}}}&=0,\label{eq:zero_overlap1}\\ \vcenter{\hbox{\includegraphics[scale=0.5]{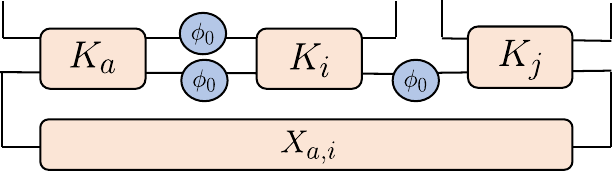}}}&=0.\label{eq:zero_overlap2}
\end{align}
From Eq.~\eqref{eq:zero_overlap1}, we obtain
\begin{align}
    \left\lVert \vcenter{\hbox{\includegraphics[scale=0.5]{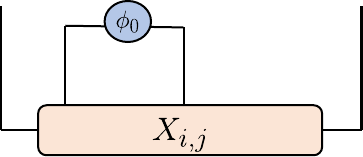}}}\right\rVert &\leq 
    \left\lVert \vcenter{\hbox{\includegraphics[scale=0.5]{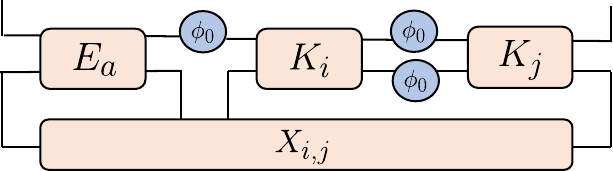}}}\right\rVert  + \left\lVert \vcenter{\hbox{\includegraphics[scale=0.5]{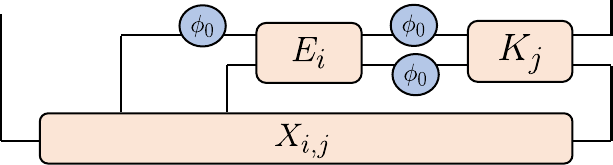}}}\right\rVert\notag + \\\notag &\qquad\left\lVert \vcenter{\hbox{\includegraphics[scale=0.5]{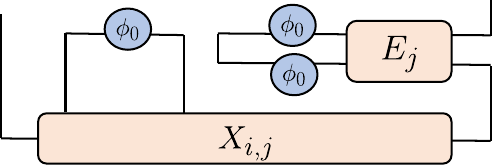}}}\right\rVert\\
    &\leq \norm{\ket{X_{i,j}}}(1+\delta)^2\delta + \norm{\ket{X_{i,j}}}(1+\delta)\delta+\norm{\ket{X_{i,j}}}\delta =\frac{19}{4}\norm{\ket{X_{i,j}}}\delta.
\end{align}
Similarly, from Eq.~\eqref{eq:zero_overlap2}, we get
\begin{align}
    \left\lVert \vcenter{\hbox{\includegraphics[scale=0.5]{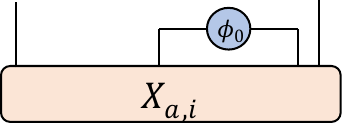}}}\right\rVert \leq 5 \norm{\ket{X_{a,i}}}\delta.
\end{align}
Next we note
\begin{align}
    \vert \braket{v}{u}\vert &\leq  \left\lVert \vcenter{\hbox{\includegraphics[scale=0.5]{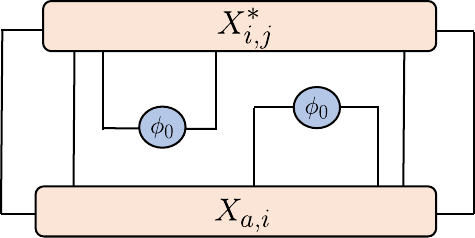}}}\right\rVert + \left\lVert \vcenter{\hbox{\includegraphics[scale=0.5]{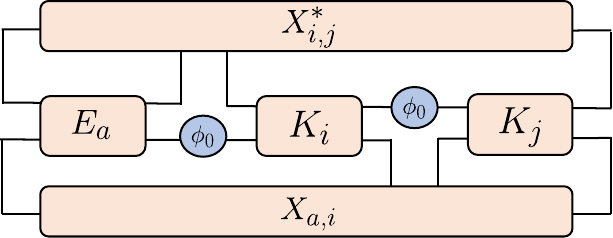}}}\right\rVert +\left\lVert \vcenter{\hbox{\includegraphics[scale=0.5]{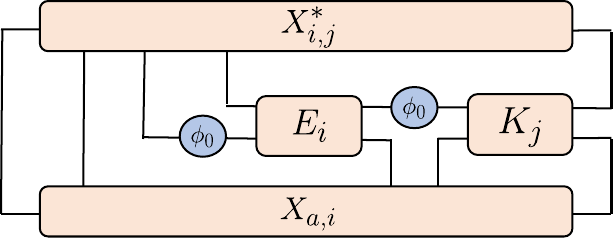}}}\right\rVert + \notag \\ &\qquad\left\lVert \vcenter{\hbox{\includegraphics[scale=0.5]{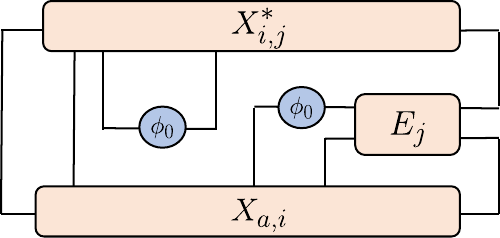}}}\right\rVert \notag 
    \\ &\leq 25 \delta^2 \norm{\ket{X_{a,i}}}\norm{\ket{X_{i,j}}}+ \norm{K_i}\norm{K_j}\norm{E_{a}}\norm{\ket{X_{i,j}}}\norm{\ket{X_{a,i}}} +\norm{B_j}\norm{E_i}\norm{\ket{X_{a,i}}}\norm{\ket{X_{i,j}}}+\norm{E_j}\norm{\ket{X_{a,i}}}\norm{\ket{X_{i,j}}}\notag\\
    &\leq (25\delta^2+(1+\delta)^2\delta +(1+\delta)\delta+\delta )\norm{\ket{X_{i,j}}}\norm{\ket{X_{a,i}}}\notag \\
    &\leq \delta (\frac{25}{2}+\frac94+\frac32 +1)\norm{\ket{X_{i,j}}}\norm{\ket{X_{a,i}}}\notag \\
    &\leq \frac{69\delta}{4}\norm{\ket{X_{i,j}}}\norm{\ket{X_{a,i}}}.
\end{align}
Finally, \begin{align}
    \norm{\ket{u}}^2&=\left\lVert \vcenter{\hbox{\includegraphics[scale=0.5]{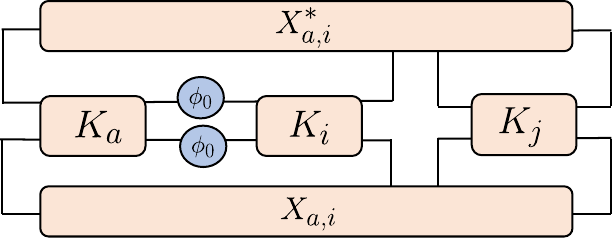}}}\right\rVert \geq (1-\delta)^3\norm{\ket{X_{a,i}}}^2  \geq \frac18 \norm{\ket{X_{a,i}}}^2.
\end{align}
Similarly, \begin{equation}
    \norm{\ket{v}}^2 \geq \frac18\norm{\ket{X_{i,j}}}^2.
\end{equation}
Thus upon combining these bounds and using Eq.~\eqref{eq:cos1},
\begin{equation}
\cos{\theta}\leq \frac{(69/4) \delta \norm{\ket{X_{a,i}}}\norm{\ket{X_{i,j}}}}{1/8 \norm{\ket{X_{a,i}}}\norm{\ket{X_{i,j}}}}\leq 138\delta.
\end{equation}
Assuming $\delta < \frac{1}{138}$, we get
\begin{align}
    \norm{A_{S_{i,j}, S_{a,i}}}, \norm{B_{S_{i,j}, S_{a,i}}} \leq \sum\limits_{n=0}^\infty (\cos{\theta})^n=\frac{1}{1-\cos{\theta}}=\frac{1}{1-138\delta}.
\end{align}
\end{proof}

\subsubsection{Mixing time analysis}\label{sec:lind_peps}
Recall the Parent Lindbladian $\mathcal{L}$ defined in Eq.~\eqref{eq:lindbladian_peps} consisting of jump operators $L_{\alpha,(i,j)}$. Suppose $P_{S_{i,j}}$ is the projector onto the ground-state subspace of $h_{i,j}$ and $P_{S_{i,j}^\perp}$ is the projector onto its orthogonal complement.
These satisfy $L_{\alpha,(i,j)}P_{S_{i,j}}=0$. Thus $L_{\alpha,(i,j)}=L_{\alpha,(i,j)}P_{S_{i,j}^\perp}$. 
\begin{lemma}[Effective Hamiltonian and Interference Bounds]
\label{lemma:heff_and_interference}
Let the effective Hamiltonian be defined as $$H_{\textnormal{eff},(i,j)} = \frac{1}{2}\sum_\alpha L^\dagger_{\alpha,(i,j)}L_{\alpha,(i,j)},$$ and define the interference operator $F_{a,i,j} := \sum_\alpha L_{\alpha,(a,i)}^\dagger [h_{i,j} , L_{\alpha,(a,i)}]$. For $\delta \le 1/2$, the following bounds hold:
\begin{enumerate}
    \item $\left\Vert H_{\textnormal{eff},(i,j)}- \frac{1}{2}h_{i,j} \right\Vert \le 5\delta$.
    \item $\left\Vert F_{a,i,j} \right\Vert \le 48\delta$.
\end{enumerate}
\end{lemma}
\begin{proof}
By definition, $H_{\text{eff},(i,j)}$ is supported exclusively on the orthogonal complement, $H_{\text{eff},(i,j)}= P_{S_{i,j}^\perp}H_{\text{eff},(i,j)} P_{S_{i,j}^\perp}$, and expands to:
\begin{equation}
\label{eq:eff_ham}
    H_{\text{eff},(i,j)} = \frac{1}{2}\sum_\alpha \vcenter{\hbox{\includegraphics[scale=0.5]{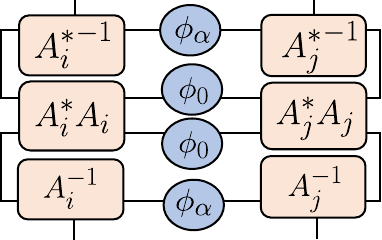}}}.
\end{equation}
Comparing this directly with Eq.~\eqref{eq:parent_ham} and defining the local tensor perturbation $E_i=A^{\dagger}_i A_i - I$, we obtain the exact difference:
\begin{equation}
    H_{\text{eff},(i,j)}- \frac{1}{2}h_{i,j} = \frac{1}{2}\sum_\alpha \left(
\vcenter{
  \hbox{
    \raisebox{-0.5\height}{\includegraphics[scale=0.5]{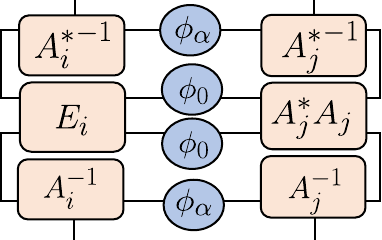}}%
    \,-\,
    \raisebox{-0.5\height}{\includegraphics[scale=0.5]{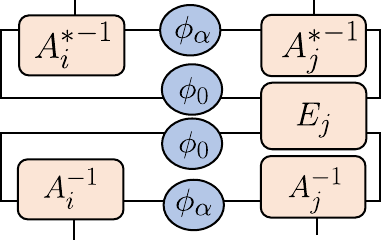}}
  }
}\right).
\end{equation}
Taking the operator norm and using $\left\Vert E_i \right\Vert \le \delta$, we bound the deviation:
\begin{align}\label{eq:Heff_norm}
    \left\Vert H_{\text{eff},(i,j)}- \frac{1}{2}h_{i,j} \right\Vert &\le \frac{1}{2}\left( (\smallnorm{A_j^\dagger A_j } +1)\left\Vert A_i^{-1} \right\Vert^2 \left\Vert A_j^{-1} \right\Vert^2 \left\Vert E_i \right\Vert \right) \notag\\
    &\le \frac{1}{2} (2+\delta)(1-\delta)^{-2}\delta \notag\\
    &\le \frac{1}{2}(2)^2\left(\frac52\right)\delta \le 5 \delta.
\end{align}

To bound the interference operator $F_{a,i,j}$, let $K_i= A_i^\dagger A_i$ and $\widetilde{E}_i=K_i^{-1} - I$. Evaluating the commutator yields:
\begingroup
\allowdisplaybreaks
\begin{align}
  F_{a,i,j} &= \sum_\alpha \left(
    \vcenter{\hbox{\includegraphics[scale=0.5]{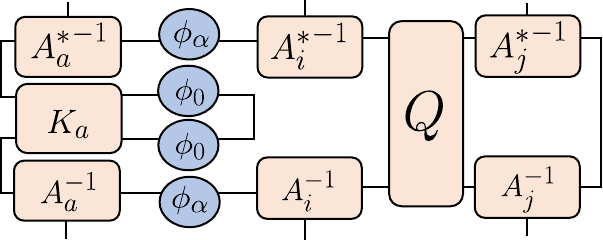}}} - \vcenter{\hbox{\includegraphics[scale=0.5]{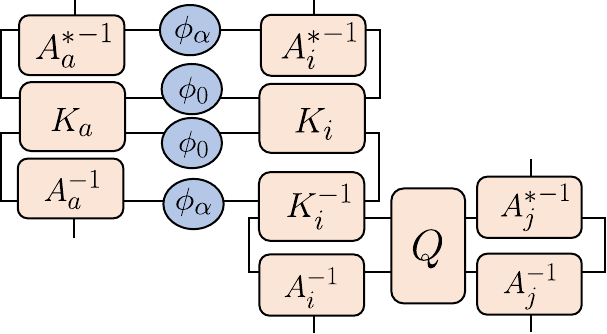}}} \right)\notag\\
  &= \sum_\alpha \left(
    \vcenter{\hbox{\includegraphics[scale=0.5]{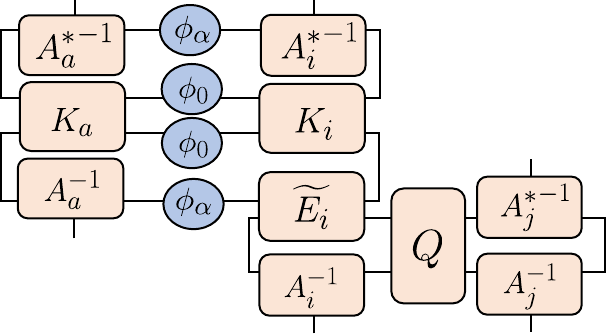}}} - \vcenter{\hbox{\includegraphics[scale=0.5]{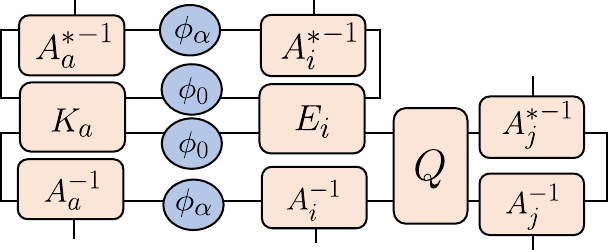}}} \right).
\end{align}
\endgroup
Since $\left\Vert E_i \right\Vert \le \delta$, the inverse perturbation is bounded by $\smallnorm{\widetilde{E}_i } \le \max\left( \frac{1}{1-\delta}-1, 1-\frac{1}{1+\delta} \right) \le \delta$. We thus bound the overall norm of the interference operator as:
\begin{align}\label{eq:Faij_bound}
    \left\Vert F_{a,i,j} \right\Vert &\le \left\Vert \sum_\alpha\braket{\phi_\alpha}{\phi_\alpha} \right\Vert \left\Vert A^{-1}_a \right\Vert^2 \left\Vert A^{-1}_i \right\Vert^2 \left\Vert A^{-1}_j \right\Vert^2 \left\Vert K_a \right\Vert \left( \left\Vert K_i \right\Vert \smallnorm{\widetilde{E}_i }+ \left\Vert E_i \right\Vert \right)\notag \\
    &\le (1-\delta)^{-3}(1+\delta)((1+\delta)2\delta+\delta) \notag\\
    &\le 8\left(\frac32\right)\left( \frac32 \cdot 2\delta+\delta\right) \le 48\delta.
\end{align}
\end{proof}
We observe that the parent Hamiltonian $H=\sum\limits_{(i,j)\in E} h_{i,j}$ of the PEPS $\ket{\psi}$ defined in Eq.~\eqref{eq:PEPS-def} has a uniform (in $N$) gap $\Delta_H$. Thus, for any state $\rho$, 
\begin{equation}
    \tr{}{\rho H}\geq \Delta_H \tr{}{(I-\ket{\psi}\!\bra{\psi})\rho},
\end{equation}
which further implies
 \begin{align}\label{eq:energy_parent_ham}
     \braket{\psi\vert \rho}{\psi}\geq 1-\frac{\tr{}{\rho H}  }{\Delta_H}.
 \end{align}

Hence, we need to analyze the expectation value of the parent Hamiltonian as a function of time, because this will allow us to examine the overlap with its ground state. Denote the state at time $t$ by \begin{equation}
    \rho(t)=e^{\mathcal{L}t}\rho(0),
\end{equation} where $\rho(0)$ is an arbitrary initial state. Also define $h_{i,j}(t)=\tr{}{h_{i,j}\rho(t)}$. We now write a Heisenberg equation of motion to analyze the Lindbladian dynamics. 
\begin{align}
    \label{eq:dynamics_parent_ham_peps}
    \frac{d}{dt} \tr{}{h_{i,j}\rho(t)} = \tr{}{\mathcal{L}_{(i,j)}^\dagger (h_{i,j})\rho(t)}+\sum\limits_{a\in \mathcal{N}_i}\tr{}{\mathcal{L}_{(a,i)}^\dagger (h_{i,j})\rho(t)} +\sum\limits_{k\in \mathcal{N}_j}\tr{}{\mathcal{L}_{(j,k)}^\dagger (h_{i,j})\rho(t)}.
\end{align}
To bound the decay of the total energy, we analyze these terms one by one by establishing bounds for the diagonal dissipation and the off-diagonal interference.

\begin{lemma}[Diagonal Dissipation Bound]\label{lemma:diag_dissipation}
For a local parent Hamiltonian term $h_{i,j}$ evolving under its corresponding local Lindbladian generator $\mathcal{L}_{(i,j)}$, the energy decays as:
\begin{align}\label{eq:term1}
    \tr{}{\mathcal{L}_{(i,j)}^\dagger (h_{i,j})\rho(t)}\leq -\left ( 1-8\delta-\frac{40\delta}{1-8\delta}\right )h_{i,j}(t).
\end{align}
\end{lemma}
\begin{proof}
We expand the Lindbladian action on the local Hamiltonian. Since the jump operators annihilate the target state, the action reduces strictly to the anti-commutator with the effective Hamiltonian:
\begin{align}
    \tr{}{\mathcal{L}_{(i,j)}^\dagger (h_{i,j})\rho(t)} &=-\tr{}{\{H_{\text{eff},(i,j)}, h_{i,j}\}\rho(t)}\notag\\
    &=-\tr{}{h_{i,j}^2\rho(t)}-\tr{}{\left\{H_{\text{eff},(i,j)}-\frac{1}{2}h_{i,j}, h_{i,j}\right\}\rho(t)}.
\end{align}
Since $h_{i,j}^2\geq \Delta_{h_{i,j}} h_{i,j} \geq(1-8\delta)h_{i,j}$ (using Lemma~\ref{lemma:local_ham_bounds}), we obtain
\begin{equation}
    \tr{}{h_{i,j}^2\rho(t)}\geq (1-8\delta)h_{i,j}(t).
\end{equation}
Since $h_{i,j}=P_{S_{i,j}^\perp}h_{i,j}P_{S_{i,j}^\perp}$ and $H_{\text{eff},(i,j)}=P_{S_{i,j}^\perp}H_{\text{eff},(i,j)}P_{S_{i,j}^\perp}$, we obtain for the anti-commutator term:
\begin{align}
    \left\vert \tr{}{\left\{ H_{\text{eff},(i,j)}-\frac{1}{2}h_{i,j},h_{i,j} \right\}\rho(t)}\right\vert 
    &=\left\vert \tr{}{\left\{ H_{\text{eff},(i,j)}-\frac{1}{2}h_{i,j},h_{i,j} \right\}P_{S_{i,j}^\perp}\rho(t)P_{S_{i,j}^\perp}}\right\vert\notag\\
    &\leq 2 \norm{h_{i,j}}\left\Vert H_{\text{eff},(i,j)}- \frac{1}{2}h_{i,j} \right\Vert\left\Vert P_{S_{i,j}^\perp}\rho(t)P_{S_{i,j}^\perp}\right\Vert_1\notag \\
    &\overset{(1)}{\leq} 2(4)(5\delta)\tr{}{P_{S_{i,j}^\perp}\rho(t)}\notag\\&\leq \frac{40\delta}{1-8\delta}h_{i,j}(t),
\end{align}
where in (1) we have used the bounds from Lemma~\ref{lemma:local_ham_bounds}, and this yields the stated bound.
\end{proof}

\begin{lemma}[Off-Diagonal Interference Bound]\label{lemma:offdiag_interference}
For all $a,i,j\in\mathcal{V}$,
\begin{align}\label{eq:Ldaggerai_2}
    \tr{}{\mathcal{L}_{(a,i)}^\dagger (h_{i,j})\rho(t)} \leq \left ( \frac{18\delta}{1-138\delta}\right )\left ( \frac{h_{i,j}(t)+3 h_{a,i}(t)}{1-8\delta} \right )+\left (\frac{48\delta}{1-8\delta} \right)h_{a,i}(t).
\end{align}
\end{lemma}
\begin{proof}
We bound the second term by explicitly expanding the action of the neighboring generator:
 \begin{align}
    \tr{}{\mathcal{L}_{(a,i)}^\dagger (h_{i,j})\rho(t)}
    &=\sum\limits_\alpha \tr{}{\left(L^\dagger_{\alpha,(a,i)}h_{i,j}L_{\alpha,(a,i)} -\frac12\{ h_{i,j}, L^\dagger_{\alpha,(a,i)}L_{\alpha,(a,i)}\}\right)\rho(t)}\notag\\
    &=\sum\limits_\alpha\left( \tr{}{L^\dagger_{\alpha,(a,i)}h_{i,j}L_{\alpha,(a,i)}P_{S^\perp_{a,i}}\rho(t)P_{S^\perp_{a,i}}}\right. -\frac12 \tr{}{h_{i,j} L^\dagger_{\alpha,(a,i)}L_{\alpha,(a,i)}P_{S^\perp_{a,i}}\rho(t)} -\notag\\ &\qquad\left.\frac12\tr{}{ L^\dagger_{\alpha,(a,i)}L_{\alpha,(a,i)}h_{i,j}\rho(t)P_{S^\perp_{a,i}}}\right)
\end{align}
We now use the relation $L_{\alpha,(a,i)}=L_{\alpha,(a,i)}P_{S^\perp_{a,i}}$ and the operator $F_{a,i,j}$ defined in Lemma~\ref{lemma:heff_and_interference} to get:
\begin{align}\label{eq:Ldaggerai_intermediate}
    \tr{}{\mathcal{L}_{(a,i)}^\dagger (h_{i,j})\rho(t)} 
    &=\frac12\tr{}{F_{a,i,j}P_{S^\perp_{a,i}}\rho(t)P_{S^\perp_{a,i}}}+\frac12\tr{}{F^\dagger_{a,i,j}P_{S^\perp_{a,i}}\rho(t)P_{S^\perp_{a,i}}} -\tr{}{h_{i,j}H_{\text{eff},(a,i)}P_{S^\perp_{a,i}}\rho(t)P_{S_{a,i}}}\notag\\&-\tr{}{H_{\text{eff},(a,i)}h_{i,j}P_{S_{a,i}}\rho(t)P_{S^\perp_{a,i}}}.
\end{align}
We use the bound on $\norm{F_{a,i,j}}$ from Lemma~\ref{lemma:heff_and_interference} to get:
\begin{align}
    \left\vert \frac12 \tr{}{F_{a,i,j}P_{S^\perp_{a,i}}\rho(t)P_{S^\perp_{a,i}}} \right\vert, \left\vert \frac12 \tr{}{F^\dagger_{a,i,j}P_{S^\perp_{a,i}}\rho(t)P_{S^\perp_{a,i}}} \right\vert \leq \frac{1}{2}\norm{F_{a,i,j}}\smallnorm{P_{S^\perp_{a,i}} \rho(t) P_{S^\perp_{a,i}} }_1\leq \frac{24\delta}{1-8\delta}h_{a,i}(t).
\end{align}
Next we consider $H_{\text{eff},(a,i)}h_{i,j}P_{S_{a,i}}$. We note that for any state $\ket{\psi}$, 
\begin{align}
   H_{\text{eff},(a,i)}h_{i,j}P_{S_{a,i}}\ket{\psi} =\frac{1}{2}\sum_\alpha\vcenter{\hbox{\includegraphics[scale=0.45]{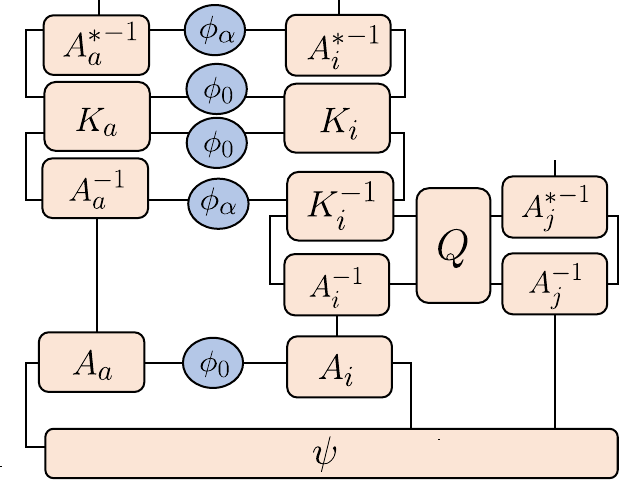}}} =\frac{1}{2}\sum_\alpha\vcenter{\hbox{\includegraphics[scale=0.45]{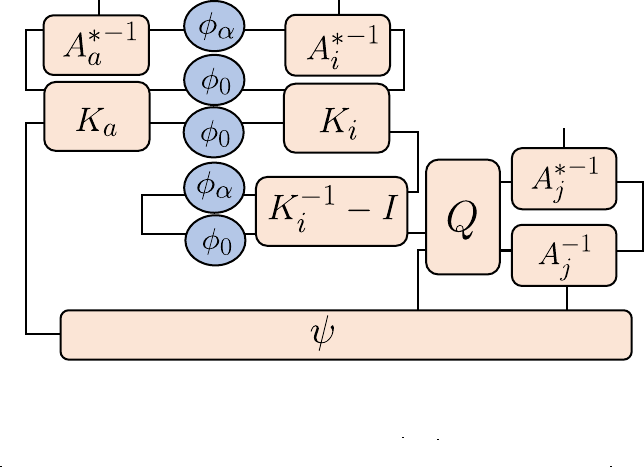}}}
\end{align}
We thus obtain that \begin{align}
    H_{\text{eff},(a,i)}h_{i,j}P_{S_{a,i}}=O_{a,i,j}P_{S_{a,i}}, \text{ where }
   O_{a,i,j} =\frac{1}{2}\sum_\alpha\vcenter{\hbox{\includegraphics[scale=0.5]{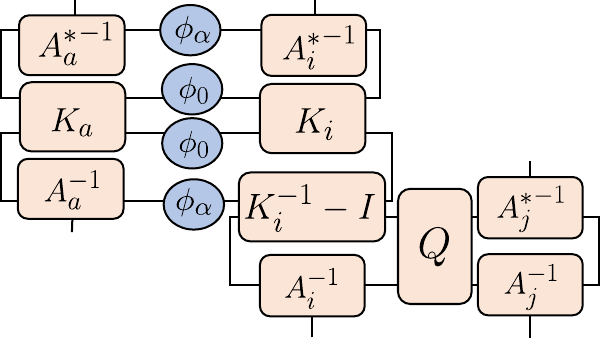}}}.
\end{align}
Here, we bound the norm of this operator as:
\begin{equation}
    \norm{O_{a,i,j}}\leq \frac{1}{2}(1-\delta)^{-3}(1+\delta)^2\norm{K^{-1}_i-I}\leq 18\delta.
\end{equation}
We now evaluate the following (where we use the relation $O_{a,i,j}P_{S_{i,j}}=0$),
\begin{align}
    \tr{}{H_{\text{eff},(a,i)}h_{i,j}P_{S_{a,i}}\rho(t)P^\perp_{S_{a,i}}} &=\tr{}{O_{a,i,j}P_{S_{a,i}}\rho(t)P^\perp_{S_{a,i}}}\notag\\
    &\overset{(1)}{=}\tr{}{O_{a,i,j}P^\perp_{S_{i,j}}P_{S_{a,i}}\rho(t)P^\perp_{S_{a,i}}}\notag\\
    &=\tr{}{O_{a,i,j}P^\perp_{S_{i,j}}A_{S_{i,j},S_{a,i}}P^\perp_{S_{i,j}}\rho(t)P^\perp_{S_{a,i}}} +\tr{}{O_{a,i,j}P^\perp_{S_{i,j}}B_{S_{i,j},S_{a,i}}P^\perp_{S_{a,i}}\rho(t)P^\perp_{S_{a,i}}},
\end{align}
where we invoke Lemma~\ref{lemma:proj} in (1).
Thus, \begin{align}
    \vert \tr{}{H_{\text{eff},(a,i)}h_{i,j}P_{S_{a,i}}\rho(t)P^\perp_{S_{a,i}}} \vert  \leq \norm{O_{a,i,j}}\norm{A_{S_{i,j},S_{a,i}}}\smallnorm{ P^\perp_{S_{i,j}} \rho(t) P^\perp_{S_{a,i}}}_1 +\norm{O_{a,i,j}}\norm{B_{S_{i,j},S_{a,i}}}\smallnorm{P^\perp_{S_{a,i}} \rho(t) P^\perp_{S_{a,i}}}_1.
\end{align}
Additionally, we use the local energies to substitute the norms:
\begin{equation}
    \label{eq:P_same} 
   \smallnorm{P^\perp_{S_{a,i}} \rho(t) P^\perp_{S_{a,i}}}_1 = \frac{h_{a,i}(t)}{1-8\delta},
\end{equation}\begin{align}\label{eq:Pdiff}
   \smallnorm{P^\perp_{S_{i,j}} \rho(t) P^\perp_{S_{a,i}}}_1 \leq \frac{(h_{i,j}(t)h_{a,i}(t))^\frac12}{1-8\delta} \leq \frac{h_{i,j}(t)+h_{a,i}(t)}{2(1-8\delta)}.
\end{align}
Thus, \begin{align}
    \vert  \tr{}{H_{\text{eff},(a,i)}h_{i,j}P_{S_{a,i}}\rho(t)P^\perp_{S_{a,i}}} \vert  
    &\leq \left(\frac{18\delta}{1-138\delta}\right)\left(\frac{h_{i,j}(t)+h_{a,i}(t)}{2(1-8\delta)}\right)+\left(\frac{18\delta}{1-138\delta}\right )\left (\frac{h_{a,i}(t)}{1-8\delta}\right )\notag\\
    &=\left(\frac{9\delta}{1-138\delta}\right)\left(\frac{h_{i,j}(t)+3h_{a,i}(t)}{1-8\delta}\right).
\end{align}
Combining these upper bounds yields the final stated inequality.
\end{proof}
Now assuming $\delta \leq \delta_\textnormal{th}$ where $\delta_\textnormal{th} = \mathcal{O}(1/\mathfrak{d}_0)$ as $\mathfrak{d}_0 \to \infty$, we prove rapid mixing of the parent Lindbladian.
\begin{proof}[Proof of Theorem 2]
Let us consider the Hamiltonian
\begin{align}
    H'=H +\sum\limits_i F_i,
\end{align}
where $H$ is the parent Hamiltonian from Eq.~\eqref{eq:parent_ham} and $F_i = P_{S_i}^\perp$ are the violation operators corresponding to each site $i$. This Hamiltonian $H'$ has the same spectral gap $\Delta_H$ as $H$, and has the PEPS $\ket{\psi}$ as its unique ground state. Therefore, combining this with Eq.~\eqref{eq:energy_parent_ham},
\begin{equation}\label{eq:mix_new}
    \left\vert \braket{\psi \vert \rho(t)}{\psi}-1\right\vert\leq \frac{\langle H'(t)\rangle}{\Delta_H}.
\end{equation}
To bound this overlap, we analyze the total time derivative $$\frac{d}{dt} \langle H'(t)\rangle = \frac{d}{dt} \langle H(t)\rangle + \frac{d}{dt} \langle F(t)\rangle.$$

Let $F=\sum\limits_i F_i$ denote the total violation operator. In the Heisenberg picture, this operator will only be acted upon by $\sum_i \mathcal{L}_i$ as the two-site Parent Lindbladian terms $\sum_e \mathcal{L}_e$ only act on states which are already in the valid subspace. Thus, the total violation operator will decay as:
\begin{equation}\label{eq:violation_decay}
    \frac{d}{dt}\langle F(t)\rangle \leq -\Gamma \langle F(t)\rangle \implies \langle F(t)\rangle \leq |E| \max_i\langle F_i(0)\rangle e^{-\Gamma t} \leq |E|e^{-\Gamma t},
\end{equation}
where $\Gamma=\min\limits_i \Gamma_i$.

For the original Parent Hamiltonian $H$, using Lemmas \ref{lemma:diag_dissipation} and \ref{lemma:offdiag_interference}, we obtain:
\begin{align}
     \mathrm{Tr}\bigg(\sum_e \mathcal{L}^\dagger_e(h_{i,j})\rho(t)\bigg)\leq-\gamma h_{i,j}(t)+\alpha \left(\sum\limits_{a\in \mathcal{N}_i}h_{a,i}(t)+\sum\limits_{k\in \mathcal{N}_j}h_{j,k}(t)\right), \text{ where}\end{align} $$
    \gamma=1-8\delta-\left(\frac{40\delta}{1-8\delta}\right)+\frac{36\delta}{(1-138\delta)(1-8\delta)},
    \alpha=\frac{54\delta}{(1-138\delta)(1-8\delta)}+\left(\frac{48\delta}{1-8\delta}\right).$$
Thus, summing over all edges yields the global dissipation rate from the two-site Lindbladians:
\begin{align}\label{eq:decay_lind_parent_ham}
    \mathrm{Tr}\bigg(\sum_e \mathcal{L}^\dagger_e(H)\rho(t)\bigg) \leq -(\gamma- 2\mathfrak{d}_0\alpha) \langle H(t)\rangle,
\end{align}
where $|\mathcal{N}_i| \leq \mathfrak{d}_0$ for all $i \in \mathcal{V}$, with $\mathfrak{d}_0$ being the maximum degree of the graph. By observing that $\gamma = 1 - \mathcal{O}(\delta)$ and $\alpha = \mathcal{O}(\delta)$, the condition $\gamma - 2\mathfrak{d}_0\alpha \geq 1/2$ is satisfied provided that $\delta$ scales inversely with the maximum degree, specifically $\delta = \mathcal{O}(\mathfrak{d}_0^{-1})$.

The time-derivative of $H$ will also have a dependence on the violation operator because of the single-site correction Lindbladians $\sum_i \mathcal{L}_i$. This preserves the same scaling as before:
\begin{align}
    \frac{d}{dt} \langle H(t)\rangle &= \mathrm{Tr}\bigg(\sum_i \mathcal{L}^\dagger_i(H)\rho(t)\bigg)+\mathrm{Tr}\bigg(\sum_e \mathcal{L}^\dagger_e(H)\rho(t)\bigg)\notag\\
    &\leq \Gamma \langle F(t)\rangle +\mathrm{Tr}\bigg(\sum_e \mathcal{L}^\dagger_e(H)\rho(t)\bigg)\notag\\
    &\overset{(1)}{\leq}\Gamma |E| e^{-\Gamma t} -  \frac12\langle H(t)\rangle,
\end{align}
where in (1) we used Eqs.~\eqref{eq:decay_lind_parent_ham} and~\eqref{eq:violation_decay}.
Upon solving we obtain $\langle H(t)\rangle\leq \mathcal{O}(|E|) e^{-\kappa t}$ for some constant $\kappa$.

Thus the total energy is bounded by $\langle H'(t)\rangle= \langle H(t)\rangle +\langle F(t)\rangle\leq \mathcal{O}(|E|) e^{-\kappa' t}$ for some constant $\kappa'$. Substituting this back into Eq.~\eqref{eq:mix_new}, we get:
\begin{equation}
    \left\vert \braket{\psi \vert \rho(t)}{\psi}-1\right\vert \leq \frac{\mathcal{O}(|E|)}{\Delta_H} e^{-\kappa' t}.
\end{equation}
Therefore, we obtain:
\begin{equation}
    \left\vert \braket{\psi \vert \rho(t)}{\psi}-1\right\vert\leq \varepsilon, \forall t\geq \tau= \mathcal{O}\left( \log\left( \frac{1}{\Delta_H\varepsilon}\right)\right)+  \mathcal{O}\left(\log\left(\frac{|E|}{\varepsilon}\right)\right).
\end{equation} 
\end{proof}
\section{Satisfying $\delta$-Unitarity for Matrix Product States (Proof of Theorem 3)}\label{sec:MPS}
    Consider a non-translationally invariant MPS $\ket{\psi}$ defined by a set of local, site-dependent tensors $\{A_i\}$,
    \begin{align}
\label{eq:mps}
    \ket{\psi}=\vcenter{\hbox{\includegraphics[scale=0.5]{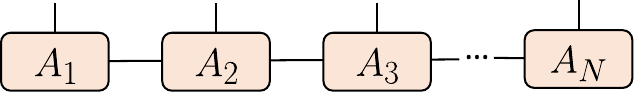}}}.
\end{align}
While we allow $A_i$ to be dependent on $i$, we will assume that $\norm{A_i}, \norm{A_i^{-1}} \leq O(1)$ independent of the system-size $N$.

We now define the blocked tensor: We construct the blocked tensor $A_n^l$, say from site $n$ to $n+l-1$, we gauge the boundaries by applying the operators $\sqrt{\rho_{n}^{-1}}$ and $\sqrt{\sigma_{n+l}^{-1}}$ to the virtual legs,
\begin{align}\label{eq:Al_blocked}
     \vcenter{\hbox{\includegraphics[scale=0.5]{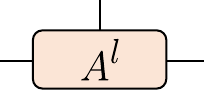}}}=\vcenter{\hbox{\includegraphics[scale=0.6]{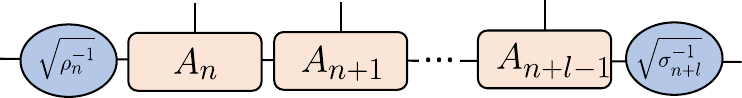}}},
\end{align}
where $\sigma_n$ and $\rho_n$ are defined as follows:
\begin{align}
    \vcenter{\hbox{\includegraphics[scale=0.5]{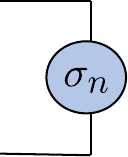}}}=\vcenter{\hbox{\includegraphics[scale=0.5]{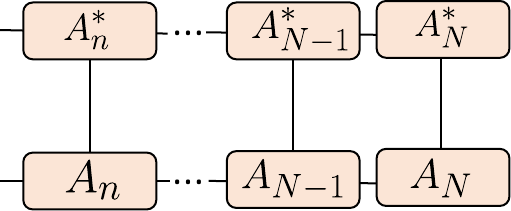}}},\qquad\vcenter{\hbox{\includegraphics[scale=0.5]{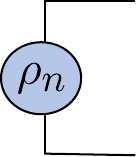}}}=\vcenter{\hbox{\includegraphics[scale=0.5]{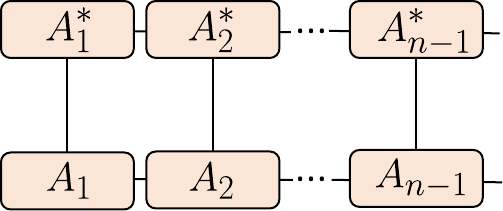}}}.
\end{align}
Clearly these satisfy $\vcenter{\hbox{\includegraphics[scale=0.5]{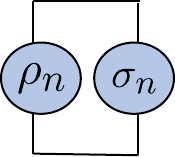}}}=\braket{\psi}{\psi}=1$. 

Next we define the notion of a state having exponentially decaying correlations~\cite{hastings2006decay}, i.e., to have a finite correlation length $\xi>0$. This means that the connected correlation function between any two local observables $O_X$ and $O_Y$ supported on disjoint finite regions $X$ and $Y$ decays exponentially with the distance $d(X, Y)$, i.e., $\exists C, \xi$ such that

\begin{equation}\label{eq:exp-decay_corr}
    |\langle \psi | O_X O_Y | \psi \rangle - \langle \psi | O_X | \psi \rangle \langle \psi | O_Y | \psi \rangle| \le C |X| |Y| e^{-d(X,Y)/\xi} \|O_X\| \cdot \|O_Y\|.
\end{equation}
\begin{lemma}\label{prop:mps}
  Let $\ket{\psi}$ be an injective Matrix Product State (MPS) $[$Eq.~\eqref{eq:mps}$]$ defined on $N$ sites, which has exponentially decaying correlations and assume $\max_{i \in \{1, 2 \dots N\}}\smallnorm{A_i^{-1}}= O(1)$ as $N \to \infty$.  Then, for any $\delta > 0$, the blocked tensors $A^l$ $[$Eq.~\eqref{eq:Al_blocked}$]$ satisfy $\lVert {A^{l}}^\dagger A^l - I \rVert \leq \delta$ on choosing the block-size $l \geq \Theta(\log(\delta^{-1})).$
\end{lemma} 
\begin{proof}
We block $l$ sites in the translationally varying MPS $\ket{\psi}$ defined in Eq.~\eqref{eq:mps}, satisfying the assumptions mentioned above and obtain blocked tensors $A^l$ as defined in Eq.~\eqref{eq:Al_blocked}.
Now let $O_X$ and $O_Y$ denote the boundary observables acting on the physical indices at sites $n-1$ and $n+l$ respectively. Let $S_X$ and $S_Y$ be arbitrary operators acting on the virtual bonds. We dress these operators with the local tensors to define our observables:
\begin{align}
    O_X=\vcenter{\hbox{\includegraphics[scale=0.5]{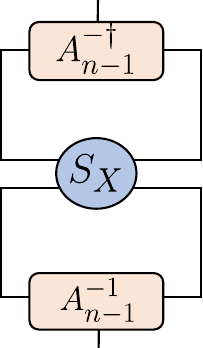}}},\qquad O_Y=\vcenter{\hbox{\includegraphics[scale=0.5]{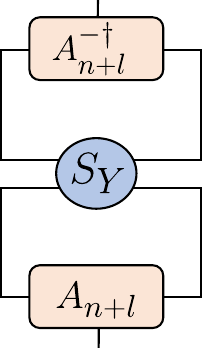}}}.
\end{align}
The operator norms of these observables are upper-bounded by:
\begin{equation}
    \smallnorm{O_X} \le \smallnorm{A_{n-1}^{-1}}^2 \smallnorm{S_X} , \quad \smallnorm{O_Y} \le \smallnorm{A_{n+l}^{-1}}^2 \smallnorm{S_Y}.
\end{equation}
We get the following expressions for the expectation values:
\begin{align}
    &\langle \psi | O_X | \psi \rangle = \vcenter{\hbox{\includegraphics[scale=0.5]{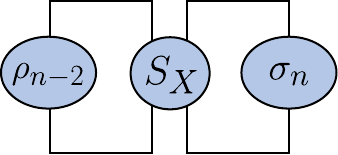}}}, \qquad \langle \psi | O_Y | \psi \rangle = \vcenter{\hbox{\includegraphics[scale=0.5]{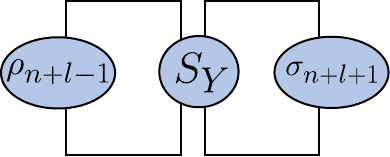}}}, \notag \\
    & \langle \psi | O_X O_Y | \psi \rangle =\vcenter{\hbox{\includegraphics[scale=0.5]{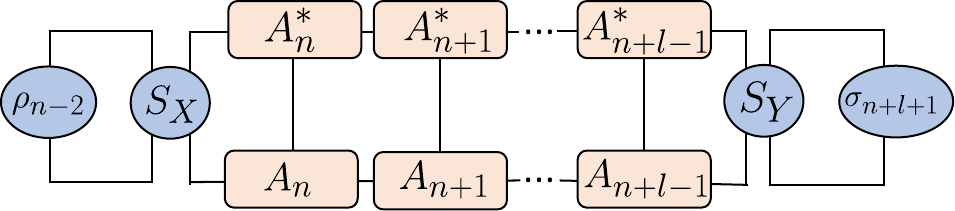}}}.
\end{align}
Now let $\mathcal{E}_{n,n+l-1}={A^l}^\dagger A^l -I.$
Substituting these boundary observables and their norm bounds into the exponential decay inequality Eq.~\eqref{eq:decay_corr} yields:
\begin{align}\label{eq:decay_corr}
    |\langle \psi | O_X O_Y | \psi \rangle - \langle \psi | O_X | \psi \rangle \langle \psi | O_Y | \psi \rangle| &=\vcenter{\hbox{\includegraphics[scale=0.5]{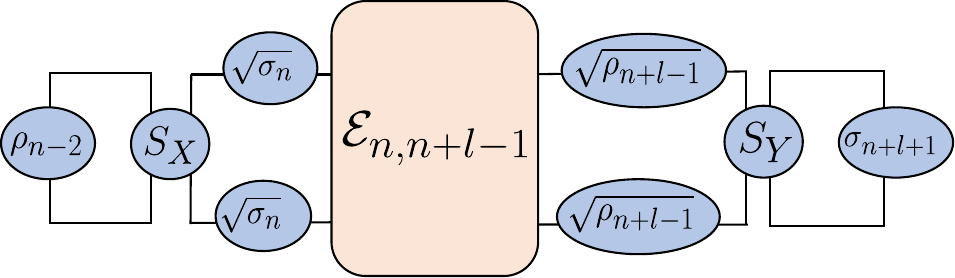}}}\notag\\&\le C e^{-l/\xi} \smallnorm{A_{n-1}^{-1}}^2 \smallnorm{A_{n+l}^{-1}}^2 \smallnorm{S_X}\smallnorm{||S_Y}.
\end{align}
We also note that
\begin{align}
    \Vert \mathcal{E}_{n,n+l-1} \Vert &\leq \sup_{F_X, F_Y} \frac{\left| \vcenter{\hbox{\includegraphics[scale=0.4]{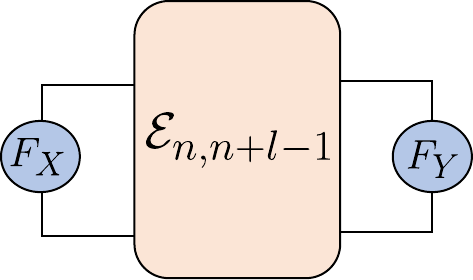}}} \right|}{\Vert F_X\Vert \Vert F_Y\Vert} \notag \\
    &= \sup_{S_X, S_Y} \frac{\left| \vcenter{\hbox{\includegraphics[scale=0.4]{E_observable.pdf}}} \right|}{\Vert S_X\Vert \Vert S_Y\Vert} \notag \\
    &\overset{(1)}{\leq} C e^{-l/\xi} \Vert A_{n-1}^{-1} \Vert^2 \Vert A_{n+l}^{-1} \Vert^2,
\end{align}
where (1) follows from \cref{eq:decay_corr}. Thus, $\lVert {A^{l}}^\dagger A^l - I \rVert \leq O(e^{-l/\xi})$. Hence for the $\delta$-isometry condition to be satisfied we need to block at least $O(\log(1/\delta))$ many sites.
\end{proof}
Combining Lemma~\ref{prop:mps} with Theorems 1 and 2, we have proved Theorem 3.
\section{Analysis of discrete-time processes (Proof of theorems 4 and 5)}\label{sec:discrete_time}
\subsection{Uniqueness of steady state}
We state a technical lemma about characterizing a certain class of quantum channels.
\begin{lemma}\label{lemma:unique_fp_channel}
    Consider channels $\mathcal{E}_i$ defined on the Hilbert space $\mathcal{H} = (\mathbb{C}^{d + 1})^{\otimes N}$ given by
    \[
    \mathcal{E}_i(X) = \sum_{\alpha = 0}^{d} K_{\alpha, i}X K_{\alpha, i}^\dagger \text{ where }K_{\alpha, i} = \sqrt{\Gamma} X \ket{0}_i \! \bra{\alpha} X^{-1} \text{ for } i \in \{1, 2 \dots d\} \text{ and }K_{0, i} = \bigg(I - \Gamma \sum_{i = 1}^d K_{\alpha, i}^\dagger K_{\alpha, i}\bigg)^{1/2},
    \]
    where $X:\mathcal{H} \to \mathcal{H}$ is an invertible operator and $\Gamma$ is chosen to be small enough such that $\Gamma \sum_{i = 1}^d K_{\alpha, i}^\dagger K_{\alpha, i} \prec I \ \forall \alpha \in \{1, 2 \dots N\}$. Partition the indices $\{1, 2 \dots N\}$ into disjoint subsets $\mathcal{I}_1, \mathcal{I}_2 \dots \mathcal{I}_k$ such that
    \[
    \forall q \in \{1, 2 \dots k\} \text{ and }\forall i, j \in \mathcal{I}_q: \mathcal{E}_i \circ \mathcal{E}_j = \mathcal{E}_j \circ \mathcal{E}_i,
    \]
    i.e., the channels corresponding to any one index set commute. Consider the channel $\mathcal{E}$ defined via
    \[
    \mathcal{E} = \frac{1}{k}\sum_{q = 1}^k \mathcal{E}_{\mathcal{I}_q} \text{ where }\mathcal{E}_{\mathcal{I}_q} = \prod_{i \in \mathcal{I}_q}\mathcal{E}_i.
    \]
    Then the channel $\mathcal{E}$ has $\ket{\psi} = X\ket{0}^{\otimes N}$ as its unique fixed point i.e., $\mathcal{E}(\rho_\infty) = \rho_\infty$ if and only if $\rho_\infty \propto \ket{\psi}\!\bra{\psi}$.
\end{lemma}
\begin{proof}
The proof of this lemma follows a strategy similar to that of lemma \ref{lemma:unique_fp_lindblad}. We begin by expressing $\mathcal{E}$ in the Krauss representation as follows:
\begin{align}
    \mathcal{E}(X) = \sum_{q = 1}^k \sum_{\vec{\alpha} \in \{0, 1, \dots d\}^{\abs{\mathcal{I}_q}}} J_{q, \vec{\alpha}} X J_{q, \vec{\alpha}}^\dagger \text{ where } J_{q, \vec{\alpha}} = \frac{1}{\sqrt{k}}\prod_{i = 1}^{\abs{\mathcal{I}_q}} K_{\alpha_i, \mathcal{I}_q(i)},
\end{align}
where $\mathcal{I}_q(i)$ is the $i^\text{th}$ element of $\mathcal{I}_q$. Note that the Kraus operators $K_{\alpha_i, \mathcal{I}_q(i)}$ and $K_{\alpha_j, \mathcal{I}_q(j)}$ commute for $i \neq j$ since the channels $\mathcal{E}_{\mathcal{I}_q(i)}$ and $\mathcal{E}_{\mathcal{I}_q(j)}$ commute. Suppose $\rho_\infty$ is a fixed point of $\mathcal{E}$ i.e., $\mathcal{E}(\rho_\infty) = \rho_\infty$: Denote by $\mathcal{S}_\infty = \{\ket{e} : \ket{e} \text{ is an eigenvector of }\rho_\infty\}$, $P_{\mathcal{S}_\infty}$ the projector on $\mathcal{S}_\infty$ and $P_{\mathcal{S}_\infty}^\perp = I - P_{\mathcal{S}_\infty}$. Then, since $P_{\mathcal{S}_\infty}^\perp \mathcal{E}(\rho_\infty)P_{\mathcal{S}_\infty}^\perp = P_{\mathcal{S}_\infty}^\perp \rho_\infty P_{\mathcal{S}_\infty}^\perp = 0$, we obtain that $P^\perp_{\mathcal{S}_\infty}J_{q, \vec{\alpha}}\rho_\infty J_{q, \vec{\alpha}}^\dagger P^\perp_{\mathcal{S}_\infty} = 0$. Thus, we obtain that
\begin{align}\label{eq:invariant_subspace_channel}
\forall \ket{\phi} \in \mathcal{S}_\infty \text{ and }\forall \vec{\alpha} \in \{0, 1 \dots d\}^{\abs{\mathcal{I}_q}}: J_{q, \vec{\alpha}} \ket{\phi} \in \mathcal{S}_\infty.
\end{align}
Next, we observe that $K_{0, i} \ket{\psi} = K_{0, i}^\dagger \ket{\psi} = \ket{\psi} \ \forall i \in \{1, 2 \dots N\}$. Thus, we obtain that $J_{q, 00 \dots 0} \ket{\psi} = \ket{\psi}$ and therefore from $\mathcal{E}(\rho_\infty) = \rho_\infty$, it follows that $\sum_{q = 1}^k \sum_{\vec{\alpha} \neq 00\dots 0}\bra{\psi} J_{q, \vec{\alpha}} \rho_\infty J_{q, \vec{\alpha}}^\dagger \ket{\psi} = 0 \implies \bra{\psi} J_{q, \vec{\alpha}} \rho_\infty J_{q, \vec{\alpha}}^\dagger \ket{\psi} = 0 \ \forall \ \vec{\alpha} \neq 00 \dots 0$. Said differently, this is equivalent to
\begin{align}\label{eq:orthogonal_subspace_channel}
\forall \ket{\phi} \in \mathcal{S}_\infty \text{ and }\forall \vec{\alpha} \neq 00\dots 0: \bra{\psi} J_{q, \vec{\alpha}} \ket{\phi} = 0.
\end{align}
Note also that from Eqs.~\eqref{eq:invariant_subspace_channel} and~\eqref{eq:orthogonal_subspace_channel}, it follows that
\begin{align}\label{eq:orthogonal_subspace_channel_multiple}
    \forall \phi \in \mathcal{S}_\infty: \langle \psi|J_{q_1, \vec{\alpha}_1} J_{q_2,  \vec{\alpha}_2} \dots J_{q_m, \vec{\alpha}_m }\ket{\phi} = 0 \text{ if }\vec{\alpha}_1 \neq 00 \dots 0.
\end{align}
We now prove that $\mathcal{S}_\infty \subseteq \text{span}(\ket{\psi})$ and therefore $\rho_\infty \propto \ket{\psi}\!\bra{\psi}$. We begin by expressing
\[
\ket{\phi} = X\sum_{\alpha_1, \alpha_2 \dots \alpha_N = 0}^d \phi_{\alpha_1, \alpha_2 \dots \alpha_N} \ket{\alpha_1, \alpha_2 \dots \alpha_N}.
\]
Now, consider a fixed $\vec{\alpha} = \{\alpha_1, \alpha_2 \dots \alpha_N\} \neq \{0,0 \dots 0\}$: We define $\vec{\alpha}_q = \{\vec{\alpha}_i : i\in \mathcal{I}_q\}$ i.e., the restriction of $\vec{\alpha}$ to $\mathcal{I}_q$. Since $\vec{\alpha} \neq \{0, 0\dots 0\}$, at least one of $\vec{\alpha}_1, \vec{\alpha}_2 \dots \vec{\alpha}_k$ is not $\{0, 0\dots 0\}$ . Without loss of generality, we assume that it is $\vec{\alpha}_1$. Then,
\begin{align}\label{eq:extracting_coefficient}
J_{1, \vec{\alpha}_1} J_{2, \vec{\alpha}_1}\dots J_{k, \vec{\alpha}_k}\ket{\phi} = \phi_{\alpha_1, \alpha_2 \dots \alpha_N} X \ket{0, 0 \dots 0}.
\end{align}
To see this, we note that if $\alpha_i = 0$, then $K_{0, i} X \ket{\alpha_1, \alpha_2 \dots \alpha_N} = X \ket{\alpha_1, \alpha_2 \dots \alpha_N}$ and if $\alpha_i \neq 0$, then $K_{\alpha_i, i} X \ket{\alpha_1, \alpha_2 \dots \alpha_N} = X\ket{\alpha_1 \dots \alpha_{i - 1}, 0, \alpha_{i + 1} \dots \alpha_N}$. Then, using the fact that $J_{q, \vec{\beta}}$ is a product of commuting operators, we obtain that if $\vec{\alpha}_q = \vec{\beta}$, then
\[
J_{q, \vec{\beta}} X \ket{\alpha_1, \alpha_2 \dots \alpha_N} = X \ket{\alpha_1', \alpha_2' \dots \alpha_N'} \text{ where }\alpha_i' = \begin{cases}
    0 & \text{ if } i \in \mathcal{I}_q, \\
    \alpha_i & \text{ if } i \notin \mathcal{I}_q.
\end{cases}
\]
A repeated application of this identity yields Eq.~\eqref{eq:extracting_coefficient}. Finally, using Eq.~\eqref{eq:orthogonal_subspace_channel_multiple} together with Eq.~\eqref{eq:extracting_coefficient} yields $\phi_{\alpha_1, \alpha_2 \dots \alpha_N} = 0$. Thus, we conclude that $\ket{\phi} = \phi_{0, 0 \dots 0} X\ket{0, 0 \dots 0} \in \text{span}(\ket{\psi})$.
\end{proof}
Now, we construct the channel to prepare the PEPS $\ket{\psi}$ in discrete-time and prove its mixing properties. We define the local channels as follows:

\begin{subequations}\label{eq:discrete_channel_peps}
    \begin{enumerate}
        \item[(i)] Defining $S_i = \text{Im}(A_i)$ as the image of the map $A_i$ defined in Eq.~\eqref{eq:tensorA}, we construct a single-site channel $\mathcal{E}_i$ for every vertex $i \in \mathcal{V}$ that acts as the identity on $S_i$ and redistributes population from the orthogonal complement into $S_i$:
        \[
            \mathcal{E}_i(\cdot) = P_{S_i} (\cdot) P_{S_i} + \frac{P_{S_i}}{\text{dim}(S_i)} \text{Tr}(P_{S_i}^\perp (\cdot)).
        \]
        
        \item[(ii)] For every edge $e = (i,j) \in \text{E}$, we construct a two-site channel $\mathcal{E}_e$ with Kraus operators $K_{0,e}, K_{1,e} \dots K_{D^2-1,e}$ where:
        \begin{align}\label{{eq:kraus_ops}}
            K_{\alpha,e} &= \begin{cases} 
            \sqrt{\Gamma} L_{\alpha,e} & \text{for } \alpha \neq 0, \\ 
            (I - 2\Gamma H_{\text{eff},e})^{1/2} & \text{for } \alpha = 0, \end{cases}
        \end{align}
        with $L_{1,e}, L_{2,e} \dots L_{D^2-1,e}$ being the jump operators from Eq.~(4) and $H_{\text{eff},e} = \frac{1}{2}\sum_\alpha L_{\alpha,e}^\dagger L_{\alpha,e}$. We remark that $\Gamma$ is chosen to be small enough such that $\Gamma H_{\text{eff},e} \preceq I$ and thus $K_{0,e}$ is well-defined.
    \end{enumerate}
\end{subequations}
Now, to construct the full channel on all the sites, as discussed in the main text, we partition the edges $\text{E}$ into disjoint subsets $\text{E}_1, \text{E}_2 \dots \text{E}_k$ such that any two edges in the same subset do not share a vertex and the number of such subsets $k \leq O(\frak{d}_0)$. For each subset $\text{E}_i$, we construct the channel $\mathcal{E}_{\text{E}_i}$ via
\[
\mathcal{E}_{\text{E}_i} = \prod_{e \in \text{E}_i} \mathcal{E}_e.
\]
We note that since no two edges in $\text{E}_i$ share the same vertex, the two-site channels $\mathcal{E}_e$ for $e \in \text{E}_i$ commute and thus can be applied simultaneously. Now, we construct the channel $\mathcal{E}$ by picking one of the channels $\mathcal{E}_{\text{E}_i}$ uniformly at random and applying them. Mathematically, this corresponds to a channel $\mathcal{E}$ given by
\begin{align}\label{eq:local_channel}
\mathcal{E} = \left (\frac{1}{k} \sum_{i = 1}^k \mathcal{E}_{\text{E}_i}\right ) \circ\prod_{i\in \mathcal{V}}\mathcal{E}_i 
\end{align}
Now we prove that $\mathcal{E}$ has the injective PEPS $\psi$ as its unique fixed point and it also has other purely imaginary eigenvalues.
\begin{proof}[Proof of Theorem 4]
The proof proceeds in three main steps: restricting the dynamics to the valid physical subspace, establishing the uniqueness of the fixed point within that subspace, and characterizing the peripheral spectrum.

Suppose $\rho_\infty$ is a steady state satisfying $\mathcal{E}(\rho_\infty) = \rho_\infty$. We evaluate the support of this state by multiplying the steady-state equation by the orthogonal projector $P_{S_i}^\perp$ and taking the trace. By construction, the single-site channel $\mathcal{E}_i$ maps any input state strictly into the valid subspace $S_i$. Therefore, for any state $\sigma$, $\text{Tr}(P_{S_i}^\perp \mathcal{E}_i(\sigma)) = 0$. 

Let $\sigma' = \prod_{i\in \mathcal{V}}\mathcal{E}_i(\rho_\infty)$. This intermediate state is fully supported within $\bigotimes_i S_i$. Next, we apply the two-site channels. Because the constituent jump operators $L_{\alpha,e}$ and the effective Hamiltonian $H_{\text{eff},e}$ act exclusively within the valid subspace ($L_{\alpha, e} P_{S_i}^\perp = P_{S_i}^\perp L_{\alpha, e} = 0$), the Kraus operators $K_{\alpha,e}$ leave the orthogonal complement invariant. Consequently, the edge channels $\mathcal{E}_e$ cannot map populations out of the valid subspace, i.e., $\text{Tr}(P_{S_i}^\perp \mathcal{E}_e(\sigma')) = 0$. Taking the trace over the full steady-state equation yields:
\begin{equation}
    \text{Tr}\left(P_{S_i}^\perp \mathcal{E}(\rho_\infty)\right) = \text{Tr}\left(P_{S_i}^\perp \rho_\infty\right) = 0.
\end{equation}
Thus, $P_{S_i}^\perp \rho_\infty P_{S_i}^\perp = 0$, i.e., the steady state $\rho_\infty$ has strictly zero support in the orthogonal complement $S_i^\perp$ for all $i \in \mathcal{V}$.

Once confined to the valid subspace, the single-site projectors act as the identity ($P_{S_i} \rho_\infty P_{S_i} = \rho_\infty$), meaning $\prod_{i\in \mathcal{V}}\mathcal{E}_i(\rho_\infty) = \rho_\infty$. The steady-state condition thus reduces to:
\begin{equation}
    \rho_\infty = \frac{1}{k} \sum_{i = 1}^k \mathcal{E}_{\textnormal{E}_i}(\rho_\infty).
\end{equation}
Next we wish to invoke Lemma~\ref{lemma:unique_fp_channel} and for that we set $\bigotimes_{e\in E} \ket{\phi_0}_e$ as the local vacuum state $\ket{0}^{\otimes N}$ from the Lemma, and let the global invertible operator $X = \bigotimes_{i \in \mathcal{V}} A_i$, then the Kraus operators take the exact form $K_{\alpha, (i,j)} = \sqrt{\Gamma}\ket{\phi_0}_{(i,j)} \bra{\phi_\alpha}X^{-1}$. Then Lemma~\ref{lemma:unique_fp_channel} implies $\mathcal{E}(\rho_\infty) = \rho_\infty$ if and only if $\rho_\infty \propto  \ket{\psi}\!\bra{\psi}$ where $\ket{\psi} = X \ket{0}^{\otimes N}=\left( \bigotimes_{i\in \mathcal{V}} A_i \right) \bigotimes_{e\in \text{E}} \ket{\phi_0}_{e}.$

Finally, we must ensure the discrete-time mapping does not support persistent oscillations. Having established that $\ket{\psi}\!\bra{\psi}$ is its unique fixed point, we directly invoke Lemma~\ref{lemma:no_imaginary}, which dictates that $\mathcal{E}$ has no other eigenvalues on the unit circle. This concludes the proof.
\end{proof}
\subsection{Rapid mixing of the discrete-time channel for near-isometric PEPS}

For the quantum channel to be well-defined, the operator inside the square root for $K_{0,e}$ must be positive semi-definite. Thus, we impose a constraint upon $\Gamma$:
\begin{equation}
    \label{eq:gamma_constraint1}
    \Gamma \leq \frac{1}{2}\left(\frac{1-\delta}{1+\delta}\right)^2.
\end{equation}
Because $\norm{H_{\textnormal{eff},(i,j)}} \leq 2(1+\delta)^2(1-\delta)^{-2}$, this constraint ensures:
\begin{equation}\label{eq:norm_X}
    \norm{2\Gamma H_{\textnormal{eff},(i,j)}}\leq 1.
\end{equation}
We define one time-step as the application of the full global channel $\mathcal{E}$, which is composed of the single-site correction layer and the randomly sampled two-site sub-layers. The number of such sub-layers $k$ is bounded by $\mathcal{O}(\mathfrak{d}_0)$, where $\mathfrak{d}_0$ is the maximum degree of the graph.

To analyze the expectation value of the parent Hamiltonian in the Heisenberg picture, we bound the action of the local quantum two-site channels $\mathcal{E}_{i,j}$ in three cases.

\begin{lemma}[Discrete Diagonal Bound - Case I]\label{lemma:discrete_diag}
The action of a local channel $\mathcal{E}_{i,j}^\dagger$ on the parent Hamiltonian term $h_{i,j}$ corresponding to the same edge is bounded by:
\begin{align}
    \tr{}{\mathcal{E}^\dagger_{i,j} (h_{i,j})\rho(t)}  \leq \left ( 1 -\Gamma\left ( 1-8\delta-\frac{40\delta}{1-8\delta}\right ) + 72 \Gamma^2\left(\frac{1}{1-8\delta}\right)\right )h_{i,j}(t).
\end{align}
\end{lemma}
\begin{proof}
Let $X_{i,j}:= 2\Gamma H_{\textnormal{eff},(i,j)}$. The dual channel action is:
\begin{align*}
    \mathcal{E}^\dagger_{i,j} (h_{i,j})&=\Gamma \sum\limits_{\alpha \neq 0}  L^\dagger_{\alpha, (i,j)} h_{i,j} L_{\alpha, (i,j)}+K^\dagger_{0, (i,j)} h_{i,j}K_{0, (i,j)}\\
    &=\sqrt{I-X_{i,j}} h_{i,j}\sqrt{I-X_{i,j}},
\end{align*}
where the first term vanishes because the jump operators annihilate the local ground state. Next, we use the Taylor expansion for the square root:
\begin{align}\label{eq:taylor}
    &\sqrt{I- X_{i,j} }= I  - \frac12 X_{i,j} + Y_{i,j}, \notag \text{ where }\\
    & \qquad Y_{i,j}= -\frac18 X_{i,j}^2 -\frac{1}{16}X_{i,j}^3 -\frac{5}{128}X_{i,j}^4-\dots, 
\end{align}
Since $\norm{X_{i,j}}\leq 1$, we have $\norm{Y_{i,j}}\leq \frac12\norm{ X_{i,j} }^2$. Substituting this into the channel action yields:
\begin{align}\label{eq:same_edge}
    &\mathcal{E}^\dagger_{i,j} (h_{i,j}) = h_{i,j}-\frac12 \{ h_{i,j}, X_{i,j} \}+ R_{i,j}(h_{i,j}),     \end{align}
where we define the remainder superoperator $R_{i,j}(O)$ as:
\begin{align}\label{eq:Rij(O)}
    R_{i,j}(O) & = \frac14 X_{i,j} O X_{i,j}+\{h_{i,j} , Y_{i,j}\}+ Y_{i,j} h_{i,j} Y_{i,j}  -\frac12( X_{i,j} h_{i,j} Y_{i,j}+Y_{i,j} h_{i,j} X_{i,j}).
\end{align}
Note that the operators $h_{i,j}, X_{i,j}, Y_{i,j}$ act entirely within the support of the projector $P_{S_{i,j}^\perp}$. Therefore:
\begin{align}\label{eq:R_rho}
    \tr{}{R_{i,j}(h_{i,j})\rho(t)}=\tr{}{R_{i,j}(h_{i,j})P_{S_{i,j}^\perp}\rho(t)P_{S_{i,j}^\perp}},
\end{align}
which implies $\vert \tr{}{R_{i,j}(h_{i,j})\rho(t)} \vert \leq \norm{R_{i,j}(h_{i,j})} \Vert P_{S_{i,j}^\perp}\rho(t)P_{S_{i,j}^\perp}\Vert_1$. 
Using the bounds $\norm{h_{i,j}} \leq 4$ and $\norm{X_{i,j}} \leq 1$, we get:
\begin{align}\label{eq:norm_Rij}
    \norm{R_{i,j}(h_{i,j})}\leq 2 \norm{h_{i,j}}\norm{ X_{i,j} }^2 \leq 72 \Gamma^2.
\end{align}   
Using $\Vert P_{S_{i,j}^\perp}\rho(t)P_{S_{i,j}^\perp}\Vert_1 \leq h_{i,j}(t)/(1-8\delta)$, we bound the remainder:
\begin{align}
    \tr{}{R_{i,j}(h_{i,j})\rho(t)}\leq 72 \Gamma^2 \left (\frac{h_{i,j}(t)}{1-8\delta}\right ).\label{eq:lemma11_1}
\end{align}
Furthermore, using Lemma~\ref{lemma:diag_dissipation}, we obtain that:
\begin{equation}
    \tr{}{-\{ h_{i,j}, \Gamma H_{\textnormal{eff},(i,j)} \}\rho(t)}\leq -\Gamma\left ( 1-8\delta-\frac{40\delta}{1-8\delta}\right )h_{i,j}(t).\label{eq:lemma11_2}
\end{equation}
Summing these contributions from Eqs.~\eqref{eq:lemma11_1} and~\eqref{eq:lemma11_2} gives the stated inequality.
\end{proof}

\begin{lemma}[Discrete Off-Diagonal Bound - Case II]\label{lemma:discrete_offdiag}
For two different edges sharing one vertex (suppose the overlapping vertex is $i$), the interference is bounded by:
\begin{align}
    \tr{}{\mathcal{E}^\dagger_{a,i} (h_{i,j})\rho(t)} \leq h_{a,i}(t)\left(\frac{12\Gamma}{1-8\delta}\right)\left( 4\delta +3\Gamma\right) + h_{i,j}(t)\left( 1+\frac{18\Gamma^2}{1-8\delta}\right).
\end{align}
\end{lemma}
\begin{proof}
Let $X_{a,i}:= 2\Gamma H_{\textnormal{eff},(a,i)}$ and let $Y_{a,i}$ denote the higher order terms as defined in \eqref{eq:taylor}.
\begin{align}
    \mathcal{E}^\dagger_{a,i} (h_{i,j})  &= \Gamma \sum\limits_{\alpha \neq 0}  L^\dagger_{\alpha, (a,i)} h_{i,j} L_{\alpha, (a,i)}+\sqrt{I-X_{a,i}} h_{i,j}\sqrt{I-X_{a,i}}\notag\\
    &=\Gamma \sum\limits_{\alpha \neq 0}  L^\dagger_{\alpha, (a,i)} h_{i,j} L_{\alpha, (a,i)}+ h_{i,j}-\frac12 \{ h_{i,j}, X_{a,i} \}+ R_{a,i}(h_{i,j}),
\end{align}
where $R_{a,i}(h_{i,j})$ is as defined in \eqref{eq:Rij(O)}.
In this case, $Y_{a,i}$ and $X_{a,i}$ preserve the subspace defined by the projector $P_{S_{a,i}^\perp}$, but $h_{i,j}$ does not, which implies:
\begin{align}
    \tr{}{R_{a,i}(h_{i,j})\rho(t)}&=\frac14 \tr{}{ X_{a,i} h_{i,j} X_{a,i}P_{S_{a,i}^\perp}\rho(t)P_{S_{a,i}^\perp}}-  \frac12\tr{}{( X_{a,i} h_{i,j} Y_{a,i}+Y_{a,i} h_{i,j} X_{a,i})P_{S_{a,i}^\perp}\rho(t)P_{S_{a,i}^\perp}}\notag\\
    &\qquad+\tr{}{Y_{a,i} h_{i,j} Y_{a,i}P_{S_{a,i}^\perp}\rho(t)P_{S_{a,i}^\perp}}+\tr{}{\{h_{i,j} , Y_{a,i}\}\rho(t)} .\label{eq:Raij}
\end{align}   
We bound the last term as follows:
\begin{align}
    \tr{}{\{h_{i,j} , Y_{a,i}\}\rho(t)} &= \tr{}{h_{i,j} Y_{a,i}P_{S_{a,i}^\perp}\rho(t)P_{S_{i,j}^\perp}}+\tr{}{Y_{a,i}h_{i,j}P_{S_{i,j}^\perp}\rho(t)P_{S_{a,i}^\perp}}\notag\\
    &\leq 2\norm{Y_{a,i}}\norm{h_{i,j}}\left\Vert P_{S_{a,i}^\perp}\rho(t)P_{S_{i,j}^\perp}\right\Vert_1 \notag \\
    &\leq 18\Gamma^2 \left(\frac{h_{a,i}(t)+h_{i,j}(t)}{1-8\delta}\right).
\end{align}
All the other terms are of the following form and satisfy:
\begin{equation}
    \vert \tr{}{O P_{S_{a,i}^\perp}\rho(t)P_{S_{a,i}^\perp}} \vert \leq \norm{O}\left\Vert P_{S_{a,i}^\perp}\rho(t)P_{S_{a,i}^\perp}\right\Vert_1.
\end{equation}
Bounding these similarly as in Case I, we obtain:
\begin{align}\label{eq:exp_Raij}
    \vert \tr{}{R_{a,i}(h_{i,j})\rho(t)}\vert  &\leq \norm{h_{i,j}}\norm{X_{a,i}}^2\left\Vert P_{S_{a,i}^\perp}\rho(t)P_{S_{a,i}^\perp}\right\Vert_1+\vert \tr{}{\{h_{i,j} , Y_{a,i}\}\rho(t)}\vert\notag\\
    &\leq \frac{18 \Gamma^2}{1-8\delta} \left (3h_{a,i}(t)+h_{i,j}(t)\right ).
\end{align}
Substituting $X_{a,i} = 2\Gamma H_{\textnormal{eff},(a,i)}$ into the first-order expansion, we use the operator $F_{a,i,j}$ defined in Lemma~\ref{lemma:heff_and_interference} to re-express the leading terms:
\begin{align}
    \sum\limits_{\alpha \neq 0}  L^\dagger_{\alpha, (a,i)} h_{i,j} L_{\alpha, (a,i)}- \{ h_{i,j}, H_{\textnormal{eff},(a,i)} \}= F_{a,i,j}+F_{a,i,j}^\dagger,
\end{align}
and this operator is supported exactly on the projected subspace of $P_{S_{a,i}^\perp}$, which enables us to use the bound from Lemma~\ref{lemma:heff_and_interference} as follows:
\begin{align}
    \tr{}{\Gamma\left(\sum\limits_{\alpha \neq 0}  L^\dagger_{\alpha, (a,i)} h_{i,j} L_{\alpha, (a,i)}- \{ h_{i,j}, H_{\textnormal{eff},(a,i)} \}\right)\rho(t)}& \leq \Gamma \norm{F_{a,i,j}+F_{a,i,j}^\dagger}\left\Vert P_{S_{a,i}^\perp}\rho(t)P_{S_{a,i}^\perp}\right\Vert_1 \notag\\
    &\leq \Gamma\frac{48\delta}{1-8\delta}h_{a,i}(t).
\end{align}
Combining these terms yields the final inequality.
\end{proof}
\begin{lemma}[Discrete Two-Edge Interference - Case III]\label{lemma:discrete_two_edge}
When two local channels $\mathcal{E}_{a,i}^\dagger$ and $\mathcal{E}_{j,k}^\dagger$ (where $a\neq i, k \neq j)$ act on $h_{i,j}$, the action is bounded by:
\begin{align}
    \tr{}{\mathcal{E}^\dagger_{a,i} \circ \mathcal{E}^\dagger_{j,k} (h_{i,j})\rho(t)} \leq (1+\alpha)h_{i,j}(t)+\beta(h_{a,i}(t)+h_{j,k}(t)),
\end{align}
with $\alpha=\mathcal{O}(\Gamma^2)$ and $\beta=\mathcal{O}(\Gamma\delta+\Gamma^2)$.
\end{lemma}
\begin{proof}
Because the edges $(a,i)$ and $(j,k)$ are vertex-disjoint, their respective Kraus operators perfectly commute. We fully expand the simultaneous action:
\begin{align}   
    \mathcal{E}^\dagger_{a,i} \circ \mathcal{E}^\dagger_{j,k} (h_{i,j}) &= \Gamma^2 \sum\limits_{\alpha, \beta \neq 0}  L^\dagger_{\alpha, (a,i)}L^\dagger_{\beta, (j,k)} h_{i,j} L_{\beta, (j,k)}L_{\alpha, (a,i)}+\Gamma\sqrt{I-X_{j,k}} \sum\limits_{\alpha \neq 0} \left( L^\dagger_{\alpha, (a,i)} h_{i,j} L_{\alpha, (a,i)}\right)\sqrt{I-X_{j,k}}+\notag\\
    &\qquad \Gamma\sqrt{I-X_{a,i}} \sum\limits_{\alpha \neq 0} \left( L^\dagger_{\alpha, (j,k)} h_{i,j} L_{\alpha, (j,k)}\right)\sqrt{I-X_{a,i}}+ \sqrt{I-X_{a,i}}\sqrt{I-X_{j,k}}h_{i,j}\sqrt{I-X_{j,k}}\sqrt{I-X_{a,i}}\notag\\
    &=h_{i,j}+\Gamma^2 \sum\limits_{\alpha, \beta \neq 0}  L^\dagger_{\alpha, (a,i)}L^\dagger_{\beta, (j,k)} h_{i,j} L_{\beta, (j,k)}L_{\alpha, (a,i)}+\qquad \Gamma \left(\sum\limits_{\alpha \neq 0} L^\dagger_{\alpha, (a,i)} h_{i,j} L_{\alpha, (a,i)} - \{ h_{i,j}, H_{\textnormal{eff},(a,i)}\}\right)+\notag\\
    &\qquad \Gamma \left(\sum\limits_{\alpha \neq 0} L^\dagger_{\alpha, (j,k)} h_{i,j} L_{\alpha, (j,k)} - \{ h_{i,j}, H_{\textnormal{eff},(j,k)}\}\right)+ R_{a,i}(h_{i,j})+R_{j,k}(h_{i,j})+\notag\\
    &\qquad \Gamma^2 \{ H_{\textnormal{eff},(a,i)}, \{h_{i,j},H_{\textnormal{eff},(j,k)}\}\}- \Gamma^2 \left\{\sum\limits_{\alpha \neq 0} L^\dagger_{\alpha, (a,i)} h_{i,j} L_{\alpha, (a,i)}, H_{\textnormal{eff},(j,k)}\right\}- \notag\\
    &\qquad \Gamma^2 \left\{\sum\limits_{\alpha \neq 0} L^\dagger_{\alpha, (j,k)} h_{i,j} L_{\alpha, (j,k)}, H_{\textnormal{eff},(a,i)}\right\} + R_{j,k}(R_{a,i}(h_{i,j}))+ \Gamma R_{j,k}\left( \sum_{\alpha \neq 0} L^\dagger_{\alpha, (a,i)} h_{i,j} L_{\alpha, (a,i)} \right)+\notag\\
    &\qquad \Gamma R_{a,i}\left( \sum_{\alpha \neq 0} L^\dagger_{\alpha, (j,k)} h_{i,j} L_{\alpha, (j,k)} \right)- \Gamma\left(R_{a,i}(\{h_{i,j},H_{\textnormal{eff},(j,k)}\})+R_{j,k}(\{h_{i,j},H_{\textnormal{eff},(a,i)}\})\right).
\end{align}

To rigorously bound the expectation value of this expansion, we group the terms by their order in $\Gamma$ and utilize the interference operators and remainder bounds derived previously in Lemma~\ref{lemma:discrete_offdiag}.

First, the zeroth-order term trivially yields the expectation value $h_{i,j}(t)$.

Second, the first-order $\Gamma$ terms exactly correspond to the interference operators $F$ for each adjacent edge. Following Lemma~\ref{lemma:discrete_offdiag}, their expectation values are bounded individually:
\begin{align}
    \text{Tr}\bigg({\Gamma\left(\sum_{\alpha\neq 0} L^\dagger_{\alpha, (a,i)} h_{i,j} L_{\alpha, (a,i)} - \{ h_{i,j}, H_{\textnormal{eff},(a,i)}\}\right)\rho(t)}\bigg) &= \tr{}{\Gamma(F_{a,i,j} + F^\dagger_{a,i,j})\rho(t)} \leq \Gamma \frac{48\delta}{1-8\delta}h_{a,i}(t), \\
    \text{Tr}\bigg({\Gamma\left(\sum_{\alpha\neq 0} L^\dagger_{\alpha, (j,k)} h_{i,j} L_{\alpha, (j,k)} - \{ h_{i,j}, H_{\textnormal{eff},(j,k)}\}\right)\rho(t)}\bigg) &= \tr{}{\Gamma(F_{j,k,i} + F^\dagger_{j,k,i})\rho(t)} \leq \Gamma \frac{48\delta}{1-8\delta}h_{j,k}(t).
\end{align}

Third, the isolated remainder terms $R_{a,i}(h_{i,j})$ and $R_{j,k}(h_{i,j})$ preserve the exact bounds derived in Eq.~\eqref{eq:exp_Raij}:
\begin{align}
    \tr{}{R_{a,i}(h_{i,j})\rho(t)} &\leq \frac{18\Gamma^2}{1-8\delta}(3h_{a,i}(t) + h_{i,j}(t)), \\
    \tr{}{R_{j,k}(h_{i,j})\rho(t)} &\leq \frac{18\Gamma^2}{1-8\delta}(3h_{j,k}(t) + h_{i,j}(t)).
\end{align}

Finally, the remaining terms in the expansion consist entirely of $\mathcal{O}(\Gamma^2)$ cross-interactions. Because all constituent operators are bounded in norm and annihilate the local ground states, applying Lemma~\ref{lemma:proj} bounds their combined expectation value by $\mathcal{O}(\Gamma^2)\big(h_{i,j}(t) + h_{a,i}(t) + h_{j,k}(t)\big)$.

Summing the bounds from all terms together yields:
\begin{equation}
    \tr{}{\mathcal{E}^\dagger_{a,i} \circ \mathcal{E}^\dagger_{j,k} (h_{i,j})\rho(t)} \leq \big(1 + \mathcal{O}(\Gamma^2)\big)h_{i,j}(t) + \left( \Gamma \frac{48\delta}{1-8\delta} + \mathcal{O}(\Gamma^2) \right)\big(h_{a,i}(t) + h_{j,k}(t)\big).
\end{equation}
Absorbing the diagonal correction into $\alpha = \mathcal{O}(\Gamma^2)$ and the off-diagonal coefficients into $\beta = \mathcal{O}(\Gamma\delta + \Gamma^2)$ directly recovers the stated inequality.
\end{proof}
Now assuming $\delta \leq \delta_\textnormal{th}$ where $\delta_\textnormal{th} = \mathcal{O}(1/\mathfrak{d}_0)$ as $\mathfrak{d}_0 \to \infty$, we prove rapid mixing of the quantum channel.
\begin{proof}[Proof of Theorem 5]
From this point onward we set $\Gamma =\Theta(\delta)$ which ensures all higher-order corrections $\mathcal{O}(\Gamma^2)$ remain subleading compared to the dissipation $\Gamma$. We put everything together to analyze the effect of one full layer of the global channel $\mathcal{E}$ on a local operator $h_{i,j}$.

Using Lemmas \ref{lemma:discrete_diag}, \ref{lemma:discrete_offdiag}, and \ref{lemma:discrete_two_edge}, we have:
\begin{align}
    \tr{}{\mathcal{E}^\dagger(h_{i,j})\rho(t)}\leq (1-c\delta+\mathcal{O}(\delta^2))h_{i,j}(t)+ \sum_{a,b \in L.C.(i,j)} \mathcal{O}(\delta^2)h_{a,b}(t),
\end{align}
where $L.C.(i,j)$ is the set of all pairs that lie in the support of the action of $\mathcal{E}$ on the edge $(i,j)$. Because the graph has maximum degree $\mathfrak{d}_0$, the number of such pairs is upper-bounded by $\mathcal{O}(\mathfrak{d}_0)$.

Summing up $h_{i,j}$ over all $(i,j)\in E$, the global parent Hamiltonian expectation value updates as:
\begin{align}\label{eq:H_update_discrete}
    \langle H(t+1) \rangle &\leq (1-c\delta +\mathcal{O}(\delta^2))\langle H(t) \rangle + \mathfrak{d}_0\mathcal{O}(\delta^2)\langle H(t) \rangle\notag \\
    & \leq (1-\kappa\delta)\langle H(t) \rangle,
\end{align}
where $\kappa$ is a strictly positive constant, provided we set the condition:
\begin{equation}
    \label{eq:delta_constraint}
    \delta < \min \left\{\frac{c'}{\mathfrak{d}_0}, \frac{1}{8}\right\},
\end{equation}
for a sufficiently small constant $c'$. 

Finally, we must account for populations outside the valid tensor image space, exactly as done in the continuous-time case. We define the modified Parent Hamiltonian $H' = H + \sum_{i} F_i$, where $F_i = P_{S_i}^\perp$, which has the same gap $\Delta_H$ and unique ground state $\psi$. The single-site correction channels $\prod_i \mathcal{E}_i$ perfectly dissipate this leakage at a strict exponential rate $\Gamma'$, guaranteeing $\langle F(t) \rangle \leq |E| e^{-\Gamma' t}$. 

Because the two-site channels $\mathcal{E}_e$ leave the orthogonal complement invariant, the coupled total energy strictly decays:
\begin{equation}
    \langle H'(t) \rangle = \langle H(t) \rangle + \langle F(t) \rangle \leq \mathcal{O}(|E|)(1-\kappa'\delta)^t,
\end{equation}
for some effective constant $\kappa'$. By the uniform spectral gap $\Delta_H$ of the parent Hamiltonian, we relate the total energy directly to the state fidelity:
\begin{equation}
    \left\vert \braket{\psi \vert \rho(t)}{\psi}-1\right\vert \leq \frac{\langle H'(t) \rangle}{\Delta_H} \leq \frac{1}{\Delta_H}\mathcal{O}(|E|)(1-\kappa'\delta)^t.
\end{equation}
To achieve a target error $\varepsilon$, we require this bound to be $\leq \varepsilon$. Solving for the required time steps $t$, we obtain:
\begin{equation}
    t \geq \mathcal{O}\left( \frac{1}{\delta} \log\left( \frac{1}{\Delta_H\varepsilon}\right)\right)+  \mathcal{O}\left( \frac{1}{\delta}\log\left(\frac{|E|}{\varepsilon}\right)\right).
\end{equation} 
Since $|E| \leq \mathfrak{d}_0 N$, this confirms the rapid mixing time and concludes the proof of Theorem~\ref{theorem:discrete}.
\end{proof}

\end{document}